\documentclass[11pt, reqno]{article}

\usepackage[a4paper, left=3cm, right=3cm, top=3cm]{geometry}
\usepackage{euscript,amscd,amsgen,amsfonts,amssymb,latexsym,graphicx,psfrag,multicol}
\usepackage{amsmath}
\usepackage{amsthm}
\usepackage{dsfont}
\usepackage[english]{babel}
\usepackage{abstract}
\usepackage{bbm}
\usepackage{hyperref}
\usepackage{nameref}
\usepackage{cleveref}
\usepackage{extarrows}
\usepackage[arrow, matrix, curve]{xy}
\usepackage{pdfpages}
\usepackage{tikz}
\usepackage{graphicx}
\usepackage{MnSymbol}
\usetikzlibrary{through}
\usepackage[numbers,sort]{natbib}
\usepackage{tabularx}

\def \bE {{\mathbb{E}}}
\def \bF {{\mathbb{F}}}
\def \bN {{\mathbb{N}}}
\def \bR {{\mathbb{R}}}

\def \cF {{\mathcal{F}}}
\def \cG {{\mathcal{G}}}

\def \cM {{\mathcal{M}}}
\def \cN {{\mathcal{N}}}

\def \cQ {{\mathcal{Q}}}

\def \cV {{\mathcal{V}}}

\def \cZ {{\mathcal{Z}}}

\def \bP {{\textbf{P}}}
\def \bQ {{\textbf{Q}}}
\def \bS {{\textbf{S}}}

\def \d {\text{d}}

\def \CPS {{\text{CPS}}}
\def \SCPS {{\text{SCPS}}}
\def \loc {{\text{loc}}}

\def \NoTC {{\text{NoTC}}}

\def \SQ {{\tilde{S}^\bQ}}

\def \hSQ {{\widehat{S}^{\widehat{\bQ}}}}
\def \hSn {{\hat{S}^n}}
\def \hQn {{\widehat{\bQ}^n}}
\def \bSn {{\bar{S}^n}}
\def \tSn {{\tilde{S}^n}}

\def \bSQ {{\bar{S}^{\bar{\bQ}}}}
\def \bbQ {{\overline{\bQ}}}

\def \bS {{\bar{S}}}

\renewcommand{\emptyset}{\font\cmsy = cmsy10 at 10pt
 \hbox{\cmsy \char 59}
}

\numberwithin{equation}{section}

\theoremstyle{plain}                    
\newtheorem{theorem}{Theorem}[section]
\newtheorem{lemma}[theorem]{Lemma}
\newtheorem{corollary}[theorem]{Corollary}
\newtheorem{proposition}[theorem]{Proposition}
\newtheorem{remark}[theorem]{Remark}

\theoremstyle{definition}
\newtheorem{definition}[theorem]{Definition}

\newtheorem{example}[theorem]{Example}

\newtheorem{assumption}{Assumption}

\theoremstyle{remark}

\newcommand*\samethanks[1][\value{footnote}]{\footnotemark[#1]}

\renewcommand\labelenumi{(\roman{enumi})}    
\renewcommand\theenumi\labelenumi                                   

\DeclareMathSymbol{\varnothing}{\mathord}{AMSb}{"3F}
\DeclareMathOperator*{\esssup}{ess\,sup}
\DeclareMathOperator*{\essinf}{ess\,inf}

\title{Asset Price Bubbles in markets with Transaction Costs}
\author{F. Biagini, T. Reitsam}
\date{\today}

\makeatletter
\def\@tocline#1#2#3#4#5#6#7{\relax
  \ifnum #1>\c@tocdepth 
  \else
    \par \addpenalty\@secpenalty\addvspace{#2}%
    \begingroup \hyphenpenalty\@M
    \@ifempty{#4}{%
      \@tempdima\csname r@tocindent\number#1\endcsname\relax
    }{%
      \@tempdima#4\relax
    }%
    \parindent\z@ \leftskip#3\relax \advance\leftskip\@tempdima\relax
    \rightskip\@pnumwidth plus4em \parfillskip-\@pnumwidth
    #5\leavevmode\hskip-\@tempdima
      \ifcase #1
       \or\or \hskip 1em \or \hskip 2em \else \hskip 3em \fi%
      #6\nobreak\relax
    \dotfill\hbox to\@pnumwidth{\@tocpagenum{#7}}\par
    \nobreak
    \endgroup
  \fi}

\newcommand*{\rom}[1]{\expandafter\@slowromancap\romannumeral #1@}
\makeatother

\author{Francesca Biagini\thanks{Workgroup Financial and Insurance Mathematics, Department of Mathematics, Ludwig-Maximilians Universit{\"a}t, Theresienstraße 39, 80333 Munich, Germany. Emails: biagini@math.lmu.de, reitsam@math.lmu.de.} \and Thomas Reitsam\samethanks[1] \thanks{The financial support of the Verein zur Versicherungswissenschaft M\"unchen e.V. is gratefully acknowledged} }
 
 \title {Asset Price Bubbles in market models with proportional transaction costs}

\begin{document} 

\maketitle

\begin{abstract}
We study asset price bubbles in market models with proportional transaction costs $\lambda\in (0,1)$ and finite time horizon $T$ in the setting of \cite{schachermayeradmissible}. By following \cite{herdegenschweizer}, we define the fundamental value $F$ of a risky asset $S$ as the price of a super-replicating portfolio for a position terminating in one unit of the asset and zero cash. We then obtain a dual representation for the fundamental value by using the super-replication theorem of \cite{schachermayersuperreplication}. We say that an asset price has a bubble if its fundamental value differs from the ask-price $(1+\lambda)S$. We investigate the impact of transaction costs on asset price bubbles and show that our model intrinsically includes the birth of a bubble. 
\end{abstract}\noindent
\textbf{Keywords:}\begin{small} financial bubbles, fundamental value, super-replication, transaction costs, consistent price systems\end{small}
\smallbreak\noindent
\smallbreak\noindent

\section{Introduction}
\label{sec:introduction}
In this paper we study \emph{financial asset price bubbles} in market models with \emph{proportional transaction costs} and finite time horizon. In the economic literature there are several contributions discussing the impact of transaction costs on the formation of asset price bubbles. It is apparent that bubbles may also appear in markets with big transaction costs, see \cite{anthony2012financial}, \cite{gerding2007} and also \cite{gerdinglaw}, \cite{levitin2011explaining}, \cite{shiller2015irrational} for the speficic case of the real estate market. Several approaches can be found in the literature to explain bubbles, like asymmetric information, see \cite{allengorton93}, \cite{allen1993finite}, heterogenous beliefs, see \cite{harrisonkreps1978}, \cite{scheinkmanoverconfidence}, and noise trading such as positive feedback activity \cite{delong1990positive}, \cite{shleifer2000inefficient}, \cite{sornette2003critical}, in combination with limits to arbitrage, see \cite{abreu2003bubbles}, \cite{delong1990noise}, \cite{shleifer1990noise}, \cite{shleifer1997limits}. In \cite{scheinkmanoverconfidence}, the authors include transaction costs in an equilibrium model with heterogeneous beliefs and show that small transaction costs may reduce speculative trading and then prevent bubble's formation. However, price volatility and size of the bubble are not reduced effectively. For an overview of heterogeneous beliefs, we refer to \cite{xiongoverview}. Also in \cite{vsperka2013transaction}, the authors show in an agent-based simulation that transaction costs can have positive impact by stabilizing the financial market model in the long run. \\
From a mathematical point of view, there is a wide literature on the theory of asset bubbles in frictionless market models. In general, a bubble is given by the difference of the \emph{market price} of the asset and its \emph{fundamental value}. While the market price can be observed, it is less obvious how to define the fundamental value. In the martingale theory of asset prices bubbles, see \cite{cox2005local}, \cite{jarrowprottercomplete}, \cite{jarrowprotterincomplete}, \cite{loewenstein2000rational}, the fundamental value of a given asset $S$ is given by its expectation of future cash flows with respect to an equivalent local martingale measure. This definition has been criticized in \cite{guasonifragility} for its sensitivity with respect to model's choice. Another approach defines the fundamental value of an asset by its super-replication prices, see \cite{herdegenschweizer}, \cite{hestonbubble}. Other models explicitly describe the impact of microeconomic interactions on asset price formation, see \cite{biagini2015formation}, \cite{jarrow2012liquidity}. In \cite{jarrow2012liquidity}, the fundamental value is exogenously given and asset price bubbles are endogenously determined by the impact of liquidity risk. In \cite{biagini2015formation}, microeconomic dynamics may at an aggregate level determine a shift in the martingale measure. Further references on asset price bubbles are \cite{biagini2014shifting}, \cite{biagini2017financial}, \cite{biagini2018liquidity}, \cite{jarrow2011detect}, \cite{jarrow2009forward}, \cite{jarrow2011foreign}, \cite{schatz2019inefficient}. For a comprehensive overview see also \cite{protterbubbles} and the entry ``Bubbles and Crashes'' of \cite{contbubbles}. \\
The aim of this paper is twofold. First, we wish to introduce and study the notion of asset price bubbles in market models with proportional transaction costs. In \cite{guasonifragility}, the authors suggest a robust definition of asset price bubbles which can be interpreted as a bubble under proportional transaction costs. However, to the best of our knowledge a thorough study of this topic is still missing in the literature.\\
Secondly, we want to investigate in a mathematical setting the impact of transaction costs on asset bubbles' formation and size. In particular, we can see that the presence of market frictions may prevent the birth of a bubble in some cases, but not always, and that we obtain that the introduction of transaction costs may not always reduce the size of a bubble, see \eqref{ali:bubblereductionexample}, consistently with the results in the economic literature.  \\
In market models with proportional transaction costs $\lambda\in (0,1)$ we distinguish between  the \emph{ask price} $(1+\lambda)S$ and the \emph{bid price} $(1-\lambda)S$ for a given asset price $S$. It is well-known that in the frictionless case the no-arbitrage condition no free lunch with vanishing risk (NFLVR) is equivalent to the existence of an equivalent local martingale measure, see \cite{delbaenschachermayer94}. In the presence of proportional transaction costs, the existence of \emph{consistent (local) price systems} (Definition \ref{def:cps}) for each $\lambda>0$ guarantees that the corresponding market model is arbitrage-free in the sense of Definition 4 of \cite{guasoniftap}. Equivalence is obtained for continuous asset price processes. Furthermore, in \cite{bayraktar2018market}, the authors establish an equivalence between the weaker notions of strictly consistent local martingale systems and the NUPBR\footnote{no unbounded profit with bounded risk} and NLABPs\footnote{no local arbitrage with bounded portfolios} conditions in the robust sense. Roughly speaking, a consistent (local) price system can be thought as a dual market model without transaction costs where the trading happens parallel. For a detailed overview of the theory of proportional transaction costs, we refer to the books \cite{kabanovbook} and \cite{schachermayerasymptotic}.\\
 Due to the presence of transaction costs, positions in cash and in the asset are asymmetric. By following \cite{herdegenschweizer}, we define the fundamental value $F$ of a given asset $S$ as the price of a super-replicating portfolio for a position terminating in one unit of the asset and zero cash. More precisely, we are interested in super-replicating a position in the asset, and not in the liquidation value of the portfolio. First we study some properties of the fundamental value. We establish a dual representation for $F$ for any time $t\in [0,T]$ based on the super-replication results from \cite{campischachermayer} and \cite{schachermayersuperreplication} and show time independence of the consistent (local) price system in the dual representation, see Theorem \ref{thm:timeindependent}. In particular, in Theorem \ref{thm:cadlag} we study when the fundamental value admits a right-continuous modification. An asset price bubble is defined as the difference of the ask price with respect to the fundamental value. For the frictionless case, if the NFLVR condition is satisfied in the setting of \cite{herdegenschweizer}, one can apply the duality result from \cite{kramkovduality} and obtain that there is a bubble if and only if the price process $S$ is a strict local martingale under all equivalent local martingale measures. In particular, if there is at least one equivalent local martingale measure such that $S$ is a true martingale, there is no bubble in the market model. Analogously, in our framework there is no bubble in the market model if there exists a consistent price system in the non-local sense for any $\lambda>0$, see Proposition \ref{lem:characterization}.\\
  Further, we discuss this theoretical setting in several examples. In particular, the impact of proportional transaction costs is investigated by comparing our model to the frictionless framework of \cite{herdegenschweizer}. It is immediate to see that no bubble in the frictionless market model means no bubble in the analogous market model with transaction costs. On the other side, if there is a bubble in the market model with proportional transaction costs, there is also a bubble in the frictionless market model. However, if there is a bubble in the frictionless market model, the introduction of transaction costs can possibly eliminate the asset price bubble. Finally, we note that our definition of asset price bubble intrinsically includes bubbles' birth, i.e., the possibility that there is no bubble at the initial time $0$, but the bubble starts at some later time $t_0>0$ with positive probability, as in the settings of \cite{biagini2014shifting} and \cite{herdegenschweizer}.\\
The paper is organized as follows. In Section \ref{sec:setting}, we outline the setting for market models with proportional transaction costs and extend the notion of admissible strategies. In Section \ref{sec:bubbles}, we introduce the definition of the fundamental value and of asset price bubbles, and establish a dual representation for the fundamental value. Further, we prove the main results, Theorem \ref{thm:timeindependent} and Theorem \ref{thm:cadlag}. In Section \ref{sec:comparison}, the impact of proportional transaction costs on bubbles' formation and size is investigated. In Section \ref{sec:examples}, we illustrate our results through concrete examples. In the Appendix, we state the super-replication results from \cite{schachermayersuperreplication} with small modifications.

\section{The Setting}
\label{sec:setting}

Let $T>0$ describe a finite time horizon and let $(\Omega,\cF,(\cF_t)_{0\leq t \leq T},\bP)$ be a filtered probability space where the filtration $\bF:=(\cF_t)_{0\leq t\leq T}$ satisfies the usual conditions of right-continuity and saturatedness, with $\cF_0=\{\emptyset,\Omega\}$ and $\cF_T=\cF$. We consider a financial market model consisting of a risk-free asset $B$, normalized to $B\equiv 1$, and a risky asset $S$. Throughout the paper we assume that $S=(S_t)_{0\leq t\leq T}$ is an $\bF$-adapted stochastic process, with c\`adl\`ag and positive paths. For trading the risky asset in the market model, proportional transaction costs $0<\lambda<1$ are charged, i.e., to buy one share of $S$ at time $t$ the trader has to pay $(1+\lambda)S_t$ and for selling one share of $S$ at time $t$ the trader receives $(1-\lambda)S_t$. The interval $[(1-\lambda)S_t,(1+\lambda)S_t]$ is called \emph{bid-ask-spread}. Further, we assume that $S_t\in L_+^1(\cF_t,\bP)$ for all $t\in [0,T]$.  This assumption is needed in the proof of Lemma \ref{lem:local} and thus also for the main result, Theorem \ref{thm:timeindependent}.

\begin{definition}
\label{def:cps}
For $\lambda>0$ and a stopping time $0\leq \sigma<\tau\leq T$, we call $\CPS(\sigma,\tau)$ (resp. $\CPS_{\loc}(\sigma,\tau)$) the family of pairs $(\bQ,\SQ)$ such that $\bQ$ is a probability measure on $\cF_\tau$, $\bQ\sim\bP|_{\cF_\tau}$, $\SQ$ is a martingale (resp. local martingale) under $\bQ$ on $\lsem \sigma,\tau\rsem$, and
\begin{align}
\label{eq:cps}
(1-\lambda)S_t\leq \SQ_t\leq (1+\lambda)S_t,\quad\text{for }\sigma\leq t\leq \tau.
\end{align}
A pair $(\bQ,\SQ)$ in $\CPS(\sigma,\tau)$ (resp. $\CPS_{\loc}(\sigma,\tau)$) is called \emph{consistent price system} (resp. \emph{consistent local price system}). By $\cQ(\sigma,T)$ (resp. $\cQ_\loc(\sigma,T)$) we denote the set of measures $\bQ$ such that there exists a pair $(\bQ,\SQ)\in\CPS(\sigma,T)$ (resp. $(\bQ,\SQ)\in\CPS_\loc(\sigma,T)$). Further, we write $L^p(\cF_\sigma,\cQ):=\bigcap_{\bQ\in\cQ(\sigma,T)} L^p(\cF_\sigma,\bQ)$ and $L^p(\cF_\sigma,\cQ_\loc):=\bigcap_{\bQ\in\cQ_\loc(\sigma,T)} L^p(\cF_\sigma,\bQ)$. By $L_+^p(\cF_\sigma,\cQ_\loc)$ (resp. $L_+^p(\cF_\sigma,\cQ)$) we denote the space of $[0,\infty)$-valued random variables $X\in L^p(\cF_\sigma,\cQ_\loc)$ (resp. $X\in L^p(\cF_\sigma,\cQ)$).
\end{definition}
\noindent
A consistent (local) price system can be thought as a frictionless market with better conditions for traders, see \cite{schachermayercps}. The existence of a $\lambda$-consistent (local) price system for every $0<\lambda<1$ implies that the corresponding market model is arbitrage-free in the sense of Definition 4 of \cite{guasoniftap}. Considering consistent price systems in the non-local or local sense corresponds in the frictionless case to the characterization of no arbitrage using true martingales or local martingales. In both cases the difference lies in the choice of admissible trading strategies. If there is no natural num\'eraire it seems reasonable to compare the portfolio with positions which may be short in each asset. On the other hand, if we fix a num\'eraire we control the portfolio only in units of the num\'eraire. In particular, we do not allow short positions in the risky asset. See Chapter 5 of \cite{guasoniftap} for a more detailed discussion. For the convenience of the reader we summarize the assumptions that we use through out the paper.

\begin{assumption}
\label{assumption}
We assume that $S$ admits a consistent \emph{local} price system for every $0<\lambda'\leq \lambda$.
\end{assumption} \noindent
In the sequel we will sometimes need a stronger assumption, namely, the existence of consistent price systems (in the non-local sense) for every $0<\lambda'\leq \lambda$.

\begin{assumption}
\label{assumption2}
We assume that $S$ admits a consistent price system for every $0<\lambda'\leq \lambda$. 
\end{assumption} \noindent

\begin{remark}
\label{rem:l1}
We first note for every $(\bQ,\SQ)\in\CPS_\loc(0,T)$, then $\SQ_t\in L_+^1(\cF_t,\bQ)$ because $\SQ$ is a $\bQ$-supermartingale.\\
Furthermore, Assumption \ref{assumption} guarantees that for any $t\in[0,T]$, $S_t\in L^1(\cF_t,\cQ_\loc)$, as \eqref{eq:cps} implies
\begin{align*}
S_t\leq \frac{1}{1-\lambda}\SQ_t,\quad \forall t\in[0,T],
\end{align*}
for any $(\bQ,\SQ)\in\CPS_\loc(0,T)$. 
\end{remark}\noindent
By following \cite{guasoniftap}, \cite{kabanov99}, for $\lambda>0$ we denote by $K_t^\lambda$ the solvency cone at time $t$, defined as
\begin{align}
K_t^\lambda=\text{cone}\left\{(1+\lambda)S_t e_1-e_2,-e_1+\frac{1}{(1-\lambda)S_t} e_2\right\},
\end{align}
where $e_1=(1,0), e_2=(0,1)$ are the unit vectors in $\bR^2$, and by $(-K_t^\lambda)^\circ$ the corresponding polar cone, given by
\begin{align}
\label{ali:polarcone}
\begin{split}
(-K_t^\lambda)^\circ&=\left\{(w_1,w_2)\in\bR_+^2\mid (1-\lambda)S_t\leq \frac{w_2}{w_1}\leq (1+\lambda)S_t\right\}\\
&=\left\{w\in\bR^2\mid \langle x,w\rangle\leq 0, \forall x\in\left( -K_t^\lambda\right)\right\}.
\end{split}
\end{align}

\begin{definition}
\label{def:representation}
We define $\mathcal{Z}(\sigma,\tau)$ (resp. $\mathcal{Z}_{\loc}(\sigma,\tau)$) as the set of processes $Z=(Z_t^1,Z_t^2)_{\sigma\leq t\leq \tau}$ such that $Z^1$ is a $\bP$-martingale on $\lsem\sigma,\tau\rsem$ and $Z^2$ is a $\bP$-martingale (resp. local $\bP$-martingale) on $\lsem\sigma,\tau\rsem$ and such that $Z_t\in (-K_t^\lambda)^\circ\backslash\{0\}$ a.s. for all $t\in\lsem\sigma,\tau\rsem$.
\end{definition} \noindent
The following proposition from \cite{guasoniftap} provides a convenient representation of consistent (local) price systems by elements in $\mathcal{Z}$ (resp. $\mathcal{Z}_\loc$) and follows directly from the definition of $(-K_t^\lambda)^\circ$ in \eqref{ali:polarcone}.

\begin{proposition}[Proposition 3, \cite{guasoniftap}]
\label{prop:pmartingale}
Let $Z=(Z_t^1,Z_t^2)_{\sigma\leq t\leq T}$ be a $2$-dimensional stochastic process with $Z_\tau^1\in L^1(\cF_\tau,\bP)$. Define the measure $\bQ(Z)\ll\bP$ by $\d\bQ(Z)/\d\bP:=Z_\tau^1/\bE[Z_\tau^1]$. Then $Z\in \cZ(\sigma,\tau)$ (resp. $Z\in \cZ_{\loc}(\sigma,\tau)$) if and only if $(\bQ(Z),(Z^2/Z^1))$ is a consistent price system (resp. consistent local price system) on $\lsem\sigma,\tau\rsem$.
\end{proposition} \noindent
Next, we introduce the notion of \emph{self-financing} strategies and \emph{admissibility}, by extending Definition 3 and 5 of \cite{schachermayeradmissible} to a general starting value.

\begin{definition}
Let $0<\lambda<1$. A \emph{self-financing trading strategy} starting with initial endowment $X_\sigma\in L_+^0(\cF_\sigma,\bP)$ is a pair of $\bF$-predictable finite variation processes $(\varphi_t^1,\varphi_t^2)_{\sigma\leq t\leq T}$ on $\lsem\sigma,T\rsem$ such that
\begin{enumerate}
\item $\varphi_\sigma^1=X_\sigma$ and $\varphi_\sigma^2=0$,
\item denoting by $\varphi_t^1=\varphi_\sigma^1+\varphi_t^{1,\uparrow}-\varphi_t^{1,\downarrow}$ and $\varphi_t^2=\varphi_t^{2,\uparrow}-\varphi_t^{2,\downarrow}$, the Jordan-Hahn decomposition of $\varphi^1$ and $\varphi^2$ into the difference of two non-decreasing processes, starting at $\varphi_\sigma^{1,\uparrow}=\varphi_\sigma^{1,\downarrow}=\varphi_\sigma^{2,\uparrow}=\varphi_\sigma^{2,\downarrow}=0$, these processes satisfy
\begin{align}
\label{ali:diffform}
\d\varphi_t^1\leq (1-\lambda)S_t\d\varphi_t^{2,\downarrow}-(1+\lambda)S_t\d\varphi_t^{2,\uparrow},\quad \sigma\leq t\leq  T.
\end{align}
\end{enumerate}
\end{definition}

\begin{definition}
\label{def:strategies}
Let $0<\lambda<1$.
\begin{enumerate}
\item Let $X_\sigma\in L_+^1(\cF_\sigma,\cQ_\loc)$. Then a self-financing trading strategy $\varphi=(\varphi^1,\varphi^2)$ is called \emph{admissible in a num\'eraire-based} sense on $\lsem\sigma,T\rsem$ with $\varphi_\sigma^1=X_\sigma$ if there is  $M_\sigma\in L_+^1(\cF_\sigma,\cQ_\loc)$ such that the liquidation value $V_\tau^{liq}$ satisfies
\begin{align}
\label{ali:admnumbased}
V_\tau^{liq}(\varphi^1,\varphi^2):=\varphi_\tau^1+\left(\varphi_\tau^2\right)^+(1-\lambda)S_\tau-\left(\varphi_\tau^2\right)^-(1+\lambda)S_\tau\geq -M_\sigma,
\end{align}
for all $\lsem\sigma,T\rsem$-valued stopping times $\tau$.
\item Let $X_\sigma\in L_+^1(\cF_\sigma,\cQ)$. Then a self-financing trading strategy $\varphi=(\varphi^1,\varphi^2)$ is called \emph{admissible in a num\'eraire-free} sense on $\lsem\sigma,T\rsem$ with $\varphi_\sigma^1=X_\sigma$ if there is $(M_\sigma^1,M_\sigma^2)\in L_+^1(\cF_\sigma,\cQ)\times L_+^\infty(\cF_\sigma,\cQ)$ such that
\begin{align}
\label{ali:admnumfree}
V_\tau^{liq}(\varphi^1,\varphi^2):=\varphi_\tau^1+\left(\varphi_\tau^2\right)^+(1-\lambda)S_\tau-\left(\varphi_\tau^2\right)^-(1+\lambda)S_\tau\geq -M_\sigma^1-M_\sigma^2S_\tau,
\end{align}
for all $\lsem \sigma,T\rsem$-valued stopping times $\tau$.
\end{enumerate}
We denote by $\cV_M(X_\sigma,\sigma,T,\lambda)$ (resp. $\cV_M^\loc(X_\sigma,\sigma,T,\lambda)$) the set of all such trading strategies in the num\'eraire-free sense (resp. num\'eraire-based sense) $\varphi$ on the interval $\lsem \sigma,T\rsem$. We also use the notation $\cV_{\sigma,T}(X_\sigma,\lambda)=\bigcup_M \cV_M(X_\sigma,\sigma,T,\lambda)$ (resp. $\cV_{X_\sigma,\sigma,T}^{\loc}(\lambda)=\bigcup_M \cV_M^\loc(X_\sigma,\sigma,T,\lambda)$).
\end{definition} \noindent
For more details on the differential form of \eqref{ali:diffform} we refer the interested reader to \cite{schachermayeradmissible}.\\
Note that both accounts, the holdings in the bond $\varphi^1$ as well as the holdings in the asset $\varphi^2$ are separately given in the definition of a trading strategy $\varphi$. Having an inequality in \eqref{ali:diffform} allows for ``throwing money away'', see \cite{schachermayeradmissible}. As it is explained in \cite{schachermayeradmissible} we could require equality in \eqref{ali:diffform} in order to express $\varphi^1$ in terms of $\varphi^2$. However, for our approach it is more convenient to specify both accounts separately.

\begin{remark}
\label{rem:initialcond}
We now discuss the definition of admissible strategies. Since we are interested in considering strategies on a random interval with non-zero initial endowment, we need to extend Definitions 3 and 5 of \cite{schachermayeradmissible}, as we now explain for the num\'eraire-based case. The argument for the num\'eraire-free setting is analogous. In a first step we consider the case of zero initial endowments. Assume that $\varphi_\sigma=(0,0)$ and $V_\tau^{liq}(\varphi)\geq - M$ for all $\lsem \sigma,T\rsem$-valued stopping times $\tau$ and a constant $M>0$. Then $\varphi$ corresponds to an admissible strategy $\psi$ on $[0,T]$ according to Definition 3 of \cite{schachermayeradmissible}, where $(\psi^1,\psi^2)\equiv (0,0)$ on $\lsem 0,\sigma\rsem$ and $\psi_t=\varphi_t$ for all $t\geq \sigma$. Conversely, any strategy $\psi$ on $[0,T]$ with $(\psi^1,\psi^2)\equiv (0,0)$ on $\lsem 0,\sigma\rsem$, which is admissible in a num\'eraire-based sense in the sense of \cite{schachermayeradmissible}, also satisfies Definition \ref{def:strategies}. Suppose now to have a non-zero initial endowment. By translation, any admissible strategy on $[0,T]$ with initial endowments corresponds to an admissible strategy on $[0,T]$ without initial endowment. This correspondence is more delicate for strategies on $\lsem \sigma,T\rsem$. Let $\varphi_\sigma=(X_\sigma,0)$ for some $X_\sigma\in L_+^1(\cF_\sigma,\cQ_\loc)$ with $V_\tau^{liq}(\varphi)\geq -M$ and define $\tilde{\varphi_t}=(\tilde{\varphi}_t^1,\tilde{\varphi}_t^2):=(\varphi_t-X_\sigma,\varphi_t)$ for all $\sigma\leq t\leq T$. Then $V_\tau^{liq}(\tilde{\varphi)}=V_\tau^{liq}(\varphi)-X_\sigma\geq -M-X_\sigma=:-M_\sigma$. Hence, it is not enough to bound the liquidation value of a strategy by a constant in order to have a one-to-one correspondence of admissible strategies with and without endowments on $\lsem \sigma,T\rsem$. Definition \ref{def:strategies} allows to obtain from any admissible strategy $\psi$ on $[0,T]$ an admissible strategy $\varphi:=\psi|_{\lsem \sigma,T\rsem}$ on $\lsem \sigma,T\rsem$. Note that in the case of $\sigma=0$ Definition \ref{def:strategies} and Definition 3 of \cite{schachermayeradmissible} coincide.\\
When we consider the definition of admissibility in a num\'eraire-based sense on $[0,T]$ from an economical perspective, the role of $M>0$ is to hedge the portfolio by $M$ units of the bond, see \cite{schachermayeradmissible}. In particular, when we superhedge a portfolio on $\lsem \sigma,T\rsem$, it seems reasonable to use the information which are available up to time $\sigma$, namely, to superhedge the portfolio by $M_\sigma$ units of the bond, where $M_\sigma\in L_+^1(\cF_\sigma,\cQ_\loc)$. 
\end{remark}\noindent
We now comment on the integrability conditions of the lower bound $M_\sigma$.

\begin{remark}
\label{rem:integrability}
We discuss the local and the non-local case separately.
In Definition 3 of \cite{schachermayeradmissible} the liquidation value of an admissible strategy in the num\'eraire-based sense $\varphi=(\varphi_t^1,\varphi_t^2)_{t\in[0,T]}$ is required to be lower bounded by a constant. This guarantees that $(\varphi_t^1+\varphi_t^2\SQ_t)_{t\in[0,T]}$ is an optional strong $\bQ$-supermartingale for all $(\bQ,\SQ)\in\CPS_\loc(0,T)$, see Proposition 2 of \cite{schachermayeradmissible}.\\
As explained in Remark \ref{rem:initialcond} we wish to extend the definitions of \cite{schachermayeradmissible} to include admissible strategies on an arbitrary interval with arbitrary initial endowment. To this propose we need to impose condition \eqref{ali:admnumbased}. However, we still obtain an arbitrage-free market model.\\
In the proof of Proposition 2 of \cite{schachermayeradmissible} the lower bound is used to apply Proposition 3.3 of \cite{anselstricker}, respectively Theorem 1 of \cite{strasser}. The conditions of these results are still fulfilled on $\lsem \sigma,T\rsem$ if \eqref{ali:admnumbased} holds, and thus $(\varphi_t^1+\varphi_t^2\SQ_t)_{\sigma\leq t\leq T}$ is an optional strong $\bQ$-supermartingale for all $(\bQ,\SQ)\in\CPS_\loc(\sigma,T)$. \\
In the non-local case, Definition 5 of \cite{schachermayeradmissible} requires that the liquidation value of an admissible strategy in the num\'eraire-free sense $\varphi=(\varphi_t^1,\varphi_t^2)_{t\in [0,T]}$ satisfies
\begin{align}
	V_\tau(\varphi^1,\varphi^2)\geq -M(1+S_\tau),
\end{align}
for all $[0,T]$-valued stopping times $\tau$. This guarantees that $(\varphi_t^1+\varphi_t^2\SQ_t)_{t\in[0,T]}$ is an optional strong $\bQ$-supermartingale for all $(\bQ,\SQ)\in\CPS(0,T)$, see Proposition 3 of \cite{schachermayeradmissible}. Following the proof of Proposition 3 of \cite{schachermayeradmissible}, we apply the following conditional version of Fatou's lemma. Let $(X_n)_{n\in\bN}$ be a sequence of real-valued random variables on $(\Omega,\cF,\bQ)$ converging almost surely to $X$ and such that the negative parts $(X_n^-)_{n\in\bN}$ are uniformly $\bQ$-integrable. Then
\begin{align*}
\bE_\bQ\left[\liminf_{n\to\infty} X_n\mid\cG\right]\leq \liminf_{n\to\infty} \bE_\bQ\left[X_n\mid\cG\right].
\end{align*}
In our case, the family $\{(\varphi_\tau^1+\varphi_\tau^2\SQ_\tau)^-:\sigma\leq \tau\leq T\}$ is uniformly $\bQ$-integrable with respect to $\bQ$ for all $(\bQ,\SQ)\in\CPS(\sigma,T)$, as we have for $\sigma\leq \tau\leq T$
\begin{align*}
\varphi_\tau^1+\varphi_\tau^2\SQ_\tau\geq V_\tau^{liq}(\varphi^1,\varphi^2)\geq -M_\sigma^1-M_\sigma^2S_\tau,
\end{align*}
because $S_\tau\leq \frac{1}{1-\lambda}\SQ_\tau$ for any $(\bQ,\SQ)\in\CPS(0,T)$ and $\SQ$ is a $\bQ$-martingale, and $(M_\sigma^1,M_\sigma^2)\in L_+^1(\cF_\sigma,\cQ)\times L_+^\infty(\cF_\sigma,\cQ)$ by assumption.
Therefore, $(\varphi_t^1+\varphi_t^2\SQ_t)_{\sigma\leq t\leq T}$ is an optional strong $\bQ$-supermartingale on $\lsem \sigma,T\rsem$ for all $(\bQ,\SQ)\in\CPS(\sigma,T)$ and all trading strategies $\varphi=(\varphi_t^1,\varphi_t^2)_{\sigma\leq t\leq T}$ are admissible in a num\'eraire-free sense.
\end{remark}

\section{Asset price bubbles under proportional transaction costs}
\label{sec:bubbles}

The notion of an \emph{asset price bubble} consists of two components, namely, the market price of an asset and its \emph{fundamental value}. We assume that the market price is given by the price process $S$. For the fundamental value of an asset, we here follow the approach of \cite{herdegenschweizer} and define the fundamental value by means of the \emph{super-replication price} of the asset.\\
 In frictionless market models, it is equivalent to hold the asset or to have the (market) value of the asset in the money market account. This symmetry fails in the presence of transaction costs. A trader who wants to buy a share of the asset at time $t\in[0,T]$ has to pay $(1+\lambda)S_t$. A trader who wants to liquidate her position in the asset at time $t\in[0,T]$ only receives $(1-\lambda)S_t$ per share of the asset. Therefore, a natural question arises. 
Which position should we super-replicate in order to obtain a reasonable definition of the fundamental value in the presence of transaction costs? 

\begin{definition}
\label{def:bubble}
\label{def:fundamental}
The \emph{fundamental value} $F=(F_t)_{t\in[0,T]}$ of an asset $S$ at time $t\in[0,T]$ in a market model with proportional transaction costs $0<\lambda<1$ is defined by
\begin{align*}
F_t:=\essinf\left\{X_t\in L_+^1(\cF_t,\bQ_\loc):\ \exists \varphi \in\cV_{t,T}^\loc(X_t,\lambda) \text{ with }\varphi_t=(X_t,0)\text{ and }\varphi_T=(0,1)\right\}.
\end{align*}
We say there is an asset price bubble in the market model with transaction costs if $\bP(F_\sigma<(1+\lambda)S_\sigma)>0$ for some stopping time $\sigma$ with values in $[0,T]$. We define the \emph{asset price bubble} as the process $\beta=(\beta_t)_{0\leq t\leq T}$ given by 
\begin{align}
\label{ali:bubbledef}
\beta_{t}:=(1+\lambda)S_t-F_t,\quad t\in[0,T].
\end{align}
\end{definition} \noindent

\begin{remark}
\label{rem:guasoni}
	In Definition 4.2 of \cite{guasonifragility}, the authors provide a robust definition of an asset price bubble, which can also be interpreted as a bubble under proportional transaction costs. A difference with respect to Definition \ref{def:bubble} lies in the chosen specification of trading strategies. In \cite{guasonifragility}, in the worst case scenario the strategy begins in cash, but the initial capital is all in stock, or analogously, the strategy ends in cash, but the trader has to deliver one share of the asset.\\
		Specifying both components of the trading strategies in
		our model allows to consider strategies starting in cash and ending in a position in the stock only.
\end{remark}

\begin{proposition}
\label{prop:fundamentalbound}
Under Assumption \ref{assumption}, we have that the fundamental value $F=(F_t)_{t\in[0,T]}$ is such that 
\begin{align*}
F_t\leq (1+\lambda)S_t,\quad t\in [0,T],
\end{align*}
and $F_t\in L^1(\cF_t,\cQ_\loc)$, $t\in[0,T]$. Moreover, the bubble $\beta=(\beta_t)_{t\in[0,T]}$ has almost surely non-negative paths.
\end{proposition} \noindent

\begin{proof}
Consider the buy and hold strategy starting at time $t\in[0,T]$. With an initial endowment $\varphi_t=((1+\lambda)S_t,0)$ it is possible to buy one share of the asset at time $t$ and keep it until time $T$. Then $F_t\leq (1+\lambda)S_t$ for all $t\in[0,T]$. Therefore, the bubble has almost surely non-negative paths. The fact that $F_\sigma \in L_+^1(\cF_\sigma,\cQ_\loc)$ follows by Remark \ref{rem:l1}.
\end{proof} \noindent
We now comment on Definition \ref{def:fundamental}, which could be interpreted as the fundamental value for the \emph{ask price}. Alternatively, we could consider to super-replicate the position $\varphi_T=((1-\lambda)S_T,0)$ which is the liquidation value of the asset $S$ at time $T$, instead, or $\varphi_T=((1+\lambda)S_T,0)$ which is the price one has to pay to buy the asset at time $T$. A trader who wants to super-replicate $((1-\lambda)S_T,0)$ is only interested in cash, namely, in the liquidation value of the asset. However, it is not possible to re-buy at $T$ a share of the asset at price $(1-\lambda)S_T$. On the other hand, a trader who wants to super-replicate $((1+\lambda)S_T,0)$ is actually interested in having the asset at $T$ in the portfolio. So, super-replicating $((1+\lambda)S_T,0)$ might be too expensive. Therefore, we consider the position $\varphi_T=(0,1)$ and its corresponding super-replication price as fundamental value. This corresponds to the price a trader is willing to pay if she had to hold the asset in her portfolio until the terminal time $T$, see \cite{jarrowasset}. Furthermore, this definition allows to model bubble birth, as in \cite{biagini2014shifting} and \cite{herdegenschweizer}, as shown in Example \ref{ex:bubblebirth}.\\

\subsection{Super-replication theorems and dual representation}

We now provide a dual representation for the fundamental value $F$, which allows to study further properties of $F$ and of the asset price bubble under transaction costs. To this purpose we extend some super-replication theorems.\\
In a frictionless market model there are well-known super-replication theorems which establish a dual representation, see e.g. \cite{quenez}, \cite{kramkovduality}. Analogously there are super-replication theorems for market models with proportional transaction costs to obtain a dual representation, see e.g. \cite{cvitanic}, \cite{kabanov99}, \cite{kabanovstricker}, \cite{kabanovlast}. We refer to the super-replication theorems of \cite{campischachermayer} and \cite{schachermayersuperreplication}. The formulations of Theorem 1.4 and Theorem 1.5 of \cite{schachermayersuperreplication} can be found in Appendix \ref{app:superreplication}.

\begin{proposition}
\label{cor:superrepcond}
Let Assumption \ref{assumption} hold. We consider an $\cF_T$-measurable contingent claim $X_T=(X_T^1,X_T^2)$ which pays $X_T^1$ many units of the bond and $X_T^2$ many units of the risky asset at time $T$. Let $X_\sigma\in L_+^1(\cF_\sigma,\cQ_\loc)$. If
\begin{align}
\label{eq:condsuperrepcond}
X_T^1-X_\sigma+\left(X_T^2\right)^+(1-\lambda)S_T-\left(X_T^2\right)^-(1+\lambda)S_T\geq -M_\sigma,
\end{align}
for some $\cF_\sigma$-measurable random variable $M_\sigma$ satisfying $\sup_{\bQ\in\cQ_\loc}\bE_\bQ[M_\sigma]<\infty$, then the following assertions are equivalent
\begin{enumerate}
\item \label{cor:item1}There is a self-financing trading strategy $\varphi$ on $\lsem \sigma,T\rsem$ with $\varphi_\sigma=(X_\sigma,0)$ and $\varphi_T=(X_T^1,X_T^2)$ which is admissible in a num\'eraire-based sense.
\item \label{cor:item2}For every $(\bQ,\SQ)\in\CPS_\loc(\sigma,T)$ we have
\begin{align}
\label{eq:superrepupperboundcond}
\bE_\bQ\left[X_T^1-X_\sigma+X_T^2\SQ_T\mid\cF_\sigma\right]\leq 0.
\end{align}
\end{enumerate}
\end{proposition}\noindent

\begin{proof}
$\ref{cor:item1}\Rightarrow\ref{cor:item2}$ It is possible to apply Proposition 2 of \cite{schachermayeradmissible} although we consider the interval $\lsem \sigma,T\rsem$ and initial endowment $\varphi_\sigma=(X_\sigma,0)$, see Remark \ref{rem:integrability}. Then $(\varphi_t^1+\varphi_t^2\SQ_t)_{\sigma\leq t\leq T}$ is an optional strong supermartingale and thus
\begin{align*}
\bE_\bQ\left[X_T^1-X_\sigma+X_T^2\SQ_T\mid\cF_\sigma\right]=\bE_\bQ\left[\varphi_T^1-X_\sigma+\varphi_T^2\SQ_T\mid\cF_\sigma\right]\leq \varphi_\sigma^1-X_\sigma+\varphi_\sigma^2\SQ_\sigma=0.
\end{align*}
$\ref{cor:item2}\Rightarrow\ref{cor:item1}$ For
\begin{align*}
\tilde{X}_T:=\left(X_T^1-X_\sigma+M_\sigma-\sup_{\bQ\in\cQ}\bE_\bQ\left[M_\sigma\right],X_T^2\right),
\end{align*}
we have
\begin{align}
\label{ali:conditionthm14}
\tilde{X}_T^1+\left(\tilde{X}_T^2\right)^+(1-\lambda)S_T-\left(\tilde{X}_T^2\right)^-(1+\lambda)S_T\geq -\sup_{\bQ\in\cQ}\bE_\bQ\left[M_\sigma\right],
\end{align}
by equation \eqref{eq:condsuperrepcond}, and for $(\bQ,\SQ)\in\CPS(\sigma,T)$ we get
\begin{align}
\label{ali:conditionthm142}
\bE_\bQ\left[\tilde{X_T}+\tilde{X}_T^2\SQ_T\right]= \bE_\bQ\left[X_T^1-X_\sigma+X_T^2\SQ_T\right]+\bE_\bQ\left[M_\sigma\right]-\sup_{\bQ\in\cQ}\bE_\bQ\left[M_\sigma\right]\leq 0,
\end{align}
by equation \eqref{eq:superrepupperboundcond}. Thus we can apply Theorem \ref{thm:superrep} (Theorem 1.4 of \cite{schachermayersuperreplication}) which yields a strategy $\tilde{\varphi}$ with $\tilde{\varphi}\equiv 0$ on $\lsem 0,\sigma\rsem$ and $\tilde{\varphi}_T=\tilde{X}_T$ which is admissible in a num\'eraire-based sense on $[0,T]$. In particular, $\varphi=(\varphi^1,\varphi^2)$ defined by $\varphi_t^1:=\tilde{\varphi}_t^1+X_\sigma-M_\sigma+\sup_{\bQ\in\cQ}\bE_\bQ\left[M_\sigma\right]$ and $\varphi_t^2:=\tilde{\varphi_t}^2$ for $\sigma\leq t\leq T$, is an admissible strategy in the num\'eraire-based sense on $\lsem \sigma,T\rsem$ according to Definition \ref{def:strategies}.
\end{proof} \noindent
From Proposition \ref{cor:superrepcond} we obtain a duality representation for the fundamental value.

\begin{proposition}
\label{prop:duality}
Let Assumption \ref{assumption} hold. Then the fundamental value $F=(F_t)_{t\in[0,T]}$ of the asset $S$ at time $t\in[0,T]$ is given by
\begin{align}
\label{ali:propduality}
F_t=\esssup_{(\bQ,\SQ)\in\CPS_\loc(t,T,\lambda)}\bE_\bQ\left[\SQ_T\mid\cF_t\right],
\end{align}
for $t\in [0,T]$.
\end{proposition} \noindent

\begin{proof}
For $X_T=(0,1)$ and $X_t\in L_+^1(\cF_t,\cQ_\loc)$, condition \eqref{eq:condsuperrepcond} is satisfied and we get by Proposition \ref{cor:superrepcond} that
\begin{align}
\label{ali:proofduality1}
\begin{split}
&\left\{X_t\in L_+^1(\cF_t,\cQ_\loc):\ \exists \varphi \in\cV_{t,T}^\loc(X_t,\lambda) \text{ with }\varphi_t=(X_t,0)\text{ and }\varphi_T=(0,1)\right\}\\
=&\left\{X_t\in L_+^1(\cF_t,\cQ_\loc):\ \bE_\bQ\left[\SQ_T\mid\cF_t\right]\leq X_t\text{, for all }(\bQ,\SQ)\in \CPS_\loc(t,T,\lambda)\right\}=:D_t.
\end{split}
\end{align}
By Definition \ref{def:fundamental} and \eqref{ali:proofduality1} we have that
\begin{align*}
F_t=\essinf D_t,\quad t\in[0,T].
\end{align*}
It is left to show that
\begin{align}
\essinf D_t=\esssup_{(\bQ,\SQ)\in\CPS_\loc(t,T,\lambda)}\bE_\bQ\left[\SQ_T\mid\cF_t\right].
\end{align}
For the first direction ``$\leq$'' we note that $\esssup_{(\bQ,\SQ)\in\CPS_\loc(t,T,\lambda)}\bE_\bQ[\SQ_T\mid\cF_t]\in D_t$, where we used that $\esssup_{(\bQ,\SQ)\in\CPS_\loc(t,T,\lambda)}\bE_\bQ[\SQ_T\mid\cF_t]\leq (1+\lambda)S_t\in L^1(\cF_t,\cQ_\loc)$.\\
For the reverse direction ``$\geq$'' we have that $\essinf D_t\geq \bE_\bQ[\SQ_T\mid\cF_t]$ for all $(\bQ,\SQ)\in\CPS_\loc(t,T,\lambda)$ which implies by the definition of the essential supremum that $$\essinf D_t\geq \esssup_{(\bQ,\SQ)\in\CPS_\loc(t,T,\lambda)}\bE_\bQ[\SQ_T\mid\cF_t].$$
\end{proof} \noindent
Note that in the above proof $t\in[0,T]$ can be replaced by a stopping time $0\leq \sigma \leq T$.\\
In Proposition \ref{prop:duality} the essential supremum is taken over the set $\CPS_\loc(t,T,\lambda)$ which depends on the initial time $t$. In contrast, if we consider the frictionless case of \cite{herdegenschweizer} and assume that Theorem 3.2 from \cite{kramkovduality} applies, the fundamental value $S_\sigma^*$ of an asset $S$ at time $\sigma$ is given by
\begin{align*}
S_\sigma^*=\esssup_{\bQ \in\cM_\loc(S)}\bE_\bQ\left[S_T\mid\cF_\sigma\right],
\end{align*}
where $\cM_\loc(S)$ denotes the set of equivalent local martingale measures for $S$. The essential supremum is taken over all equivalent local martingale measure of $S$, independently of the initial time $\sigma$. We now show that a similar independence property also holds for the fundamental value under transaction costs, see Theorem \ref{thm:timeindependent}. In order to prove it, we need some preliminary results. We start with a local version of Lemma 6 and Corollary 3 of \cite{guasoniftap}. 

\begin{lemma}
\label{lem:local}
Let Assumption \ref{assumption} hold. For each stopping time $0\leq \sigma\leq T$ and each random variable $f\in L^1(\cF_\sigma,\bP)$ such that
\begin{align}
\label{eq:lemma61}
(1-\lambda)S_\sigma<f<(1+\lambda)S_\sigma,
\end{align}
and for each $\bar{\lambda}>\lambda$ there is an $\bar{\lambda}$-consistent local price system $(\check{\bQ},\check{S})\in\CPS_\loc(0,T,\bar{\lambda})$ with $\check{S}_\sigma=f$.
\end{lemma} \noindent

\begin{proof}
The proof is partially based\footnote{The main difference with respect to the proof of Lemma 6 of \cite{guasoniftap} is that we cannot use the martingale property of consistent price systems as in \cite{guasoniftap}, because we are now in the local setting. Hence we need some further technicalities.} on the proof of Lemma 6 of \cite{guasoniftap}. Consider the sequence of stopping time $(\tau_n)_{n\in\bN}$, where
\begin{align*}
\tau_n(\omega):=\inf\{t\geq 0\mid S_t(\omega)\geq n\}\wedge T.
\end{align*}
Note that $(\tau_n)_{n\in\bN}$ defines a localizing sequence for all $\lambda$-consistent local price systems as for $(\bQ,\SQ)\in \CPS(0,T,\lambda)$ we have 
\begin{align}
\label{ali:l1}
\SQ_t\leq (1+\lambda)S_t\leq(1+\lambda) n,
\end{align}
for all $0\leq t<\tau_n$ and hence Proposition 6.1 of \cite{schachermayersuperreplication} can be applied. Further, it holds that $\tau_n\uparrow T$ a.s. Fix $\bar{\lambda}>\lambda$. First we consider the interval $\lsem 0,\sigma\rsem$. Choose $\delta\leq\lambda$ such that $\delta+(1+\delta)(\lambda+\delta)/(1-\delta)<\bar{\lambda}$. By Assumption \ref{assumption} there exists $(\bQ(\delta),\tilde{S}(\delta))\in\CPS_\loc(0,\sigma,\delta)$ a $\delta$-consistent local price system on the interval $\lsem 0,\sigma\rsem$. We have
\begin{align}
\label{ali:prooflocal}
1-\delta\leq \frac{\tilde{S}_{\tau_n\wedge\sigma}(\delta)}{S_{\tau_n\wedge\sigma}}\leq 1+\delta.
\end{align}
For $n\in\bN$ define
\begin{align*}
f_n:=\begin{cases}
f\quad&\text{on }\{\tau_n\geq \sigma\},\\
\tilde{S}(\delta)_{\tau_n}\quad&\text{on }\{\tau_n<\sigma\}.
\end{cases}
\end{align*}
Hence by \eqref{ali:l1} we get $f_n\in L^1(\cF_{\tau_n\wedge\sigma},\bP)$, and for $h_n:=f_n/S_{\tau_n\wedge\sigma}$ we have
\begin{align*}
1-\lambda<h_n<1+\lambda,
\end{align*}
and
\begin{align*}
\left|\tilde{S}_{\tau_n\wedge\sigma}(\delta)-f_n\right|<(\lambda+\delta)S_{\tau_n\wedge\sigma}\leq \frac{\lambda+\delta}{1-\delta}\tilde{S}_{\tau_n\wedge\sigma}(\delta).
\end{align*}
This implies that $f_n\in L^1(\cF_{\tau_n\wedge\sigma},\bQ(\delta))$ as well as $f\in L^1(\cF_\sigma,\bQ(\delta))$ by \eqref{eq:lemma61} and the fact that 
\begin{align*}
f\leq (1+\lambda)S_\sigma\leq \frac{1+\lambda}{1-\lambda}\tilde{S}_\sigma(\delta).
\end{align*}
Consequently, for $\bSn_\rho:=\bE_{\bQ(\delta)}[f_n\mid\cF_\rho]$ and a stopping time $\rho$ with $0\leq \rho\leq (\tau_n\wedge\sigma)$,
\begin{align*}
\left|\bE_{\bQ(\delta)}\left[f_n\mid\cF_\rho\right]-\bE_{\bQ(\delta)}\left[\tilde{S}_{\tau_n\wedge\sigma}(\delta)\mid\cF_\rho\right]\right|
<\tilde{S}_{\rho}(\delta)\frac{\lambda+\delta}{1-\delta}\leq S_\rho\frac{(\lambda+\delta)(1+\delta)}{1-\delta},
\end{align*}
thus using \eqref{ali:prooflocal} we get
\begin{align}
\label{ali:prooflocal1}
(1-\bar{\lambda})S_\rho<\bSn_\rho<(1+\bar{\lambda})S_\rho.
\end{align}
We show that $\bSn_{\rho}$ converges almost surely to a random variable $\bS_\rho^{\bQ(\delta)}$ for all $0\leq \rho\leq \sigma$. We have that
\begin{align*}
\bE_{\bQ(\delta)}\left[f_n\mid\cF_\rho\right]=\bE_{\bQ(\delta)}\left[\mathds{1}_{\{\tau_n\geq \sigma\}}f\mid\cF_\rho\right]+\bE_{\bQ(\delta)}\left[\mathds{1}_{\{\tau_n<\sigma\}}\tilde{S}_{\tau_n}(\delta)\mid\cF_\rho\right].
\end{align*}
By the Theorem of Monotone Convergence it follows that
\begin{align}
\label{ali:conv3}
\bE_{\bQ(\delta)}\left[\mathds{1}_{\{\tau_n\geq \sigma\}}f\mid\cF_\rho\right]\overset{a.s.}{\longrightarrow}\bE_{\bQ(\delta)}\left[f\mid\cF_\rho\right].
\end{align}
For the second term we have
\begin{align*}
\bE_{\bQ(\delta)}\left[\mathds{1}_{\{\tau_n<\sigma\}}\tilde{S}_{\tau_n}(\delta)\mid\cF_\rho\right]&=\bE_{\bQ(\delta)}\left[\tilde{S}_{\tau_n\wedge \sigma}(\delta)\mid\cF_\rho\right]-\bE_{\bQ(\delta)}\left[\mathds{1}_{\{\tau_n\geq\sigma\}}\tilde{S}_{\tau_n\wedge\sigma}(\delta)\mid\cF_\rho\right]\\
&=\tilde{S}_{\tau_n\wedge \rho}(\delta)-\bE_{\bQ(\delta)}\left[\mathds{1}_{\{\tau_n\geq\sigma\}}\tilde{S}_{\sigma}(\delta)\mid\cF_\rho\right].
\end{align*}
Since $\mathds{1}_{\{\tau_n\geq \sigma\}}\leq \mathds{1}_{\{\tau_{n+1}\geq \sigma\}}$ for all $n\in\bN$, we can apply the Theorem of Monotone Convergence to conclude
\begin{align}
\label{ali:limit}
\bSn_\rho\overset{a.s.}{\longrightarrow} \bE_{\bQ(\delta)}\left[f\mid\cF_\rho\right]+\tilde{S}_\rho(\delta)-\bE_{\bQ(\delta)}\left[\tilde{S}_\sigma(\delta)\mid\cF_\rho\right]=:\bar{S}_\rho^{\bQ(\delta)}.
\end{align}
We define the process $\bar{S}^{\bQ(\delta)}=(\bar{S}_t^{\bQ(\delta)})_{0\leq t\leq \sigma}$ by \eqref{ali:limit}. Since $$\left(\bE_{\bQ(\delta)}\left[f\mid\cF_t\right]-\bE_{\bQ(\delta)}\left[\tilde{S}_\sigma(\delta)\mid\cF_t\right]\right)_{0\leq t\leq \sigma}$$ is a $\bQ(\delta)$-martingale, it admits a unique c\`adl\`ag modification. Further, $\tilde{S}(\delta)$ has a unique c\`adl\`ag modification. Therefore, $\bar{S}^{\bQ(\delta)}$ admits a unique c\`adl\`ag modification as a local $\bQ(\delta)$-martingale. By \eqref{ali:prooflocal1} $\bar{S}^{\bQ(\delta)}$ lies in the bid-ask spread for $\bar{\lambda}$ by construction. Thus $(\bQ(\delta),\bar{S}^{\bQ(\delta)})$ is a $\bar{\lambda}$-consistent local price system on $\lsem 0,\sigma\rsem$ satisfying $\bar{S}_\sigma^{\bQ(\delta)}=f$. With the same construction as in the proof of Lemma 6 of \cite{guasoniftap} we can now show the existence of a consistent local price system $(\widehat{\bQ},\widehat{S})\in\CPS_\loc(\sigma,T,\bar{\lambda})$ such that $\widehat{S}_\sigma=f$. We refer to \cite{thesis} for further details. We use this result to extend $(\bQ(\delta),\bar{S}^{\bQ(\delta)})$ to a consistent local price system $(\check{\bQ},\check{S})\in\CPS_\loc(0,T,\bar{\lambda})$ on the entire interval $[0,T]$.\\
%
%
We now define $(\check{\bQ},\check{S})\in\CPS_\loc(0,T,\bar{\lambda}))$ which satisfies $\check{S}_\sigma=f$. Set $\check{\bQ}$ by 
\begin{align*}
\frac{\d\check{\bQ}}{\d\bP}:=\frac{\frac{\d\bQ(\delta)}{\d\bP}}{\bE_\bP\left[\frac{\d\widehat{\bQ}}{\d\bP}\mid\cF_\sigma\right]}\frac{\d\widehat{\bQ}}{\d\bP},
\end{align*}
and
\begin{align*}
\check{S}_t:=\begin{cases}
\bar{S}^{\bQ(\delta)}_t,\quad&\text{ for }0\leq t\leq \sigma\\ \widehat{S}_t,\quad&\text{ for }\sigma\leq t\leq T.
\end{cases}
\end{align*}
Then $(\check{\bQ},\check{S})\in\CPS_\loc(0,T,\bar{\lambda})$ and $\check{S}_\sigma=\widehat{S}_\sigma=\bar{S}^{\bQ(\delta)}_\sigma=f$.
\end{proof} \noindent

\begin{remark}
Note that in the case of a consistent price system in the non-local sense, Lemma \ref{lem:local} coincides with Lemma 6 of \cite{guasoniftap}.
\end{remark}\noindent
The following Corollary \ref{cor:localcps} can be proved in the same way as Corollary 3 in \cite{guasoniftap} because the construction does not use the martingale property of the consistent price systems.

\begin{corollary}
\label{cor:localcps}
Let Assumption \ref{assumption} hold. For any stopping time $0\leq \sigma\leq T$, probability measure $\bbQ\sim\bP|_{\cF_\sigma}$ on $\cF_\sigma$ and random variable $f\in L^1(\cF_\sigma,\bP)$ with
\begin{align*}
(1-\lambda)S_\sigma\leq f\leq (1+\lambda)S_\sigma,
\end{align*}
there exists a $\lambda$-consistent local price system $(\bQ,\SQ)\in\CPS_\loc(\sigma,T,\lambda)$ such that $\SQ_\sigma=f$ and $\bQ|_{\cF_\sigma}=\bbQ$.
\end{corollary}\noindent
We can now show time independence for the essential supremum in the definition of the fundamental value.

\begin{theorem}
\label{thm:timeindependent}
Under Assumption \ref{assumption} the following identity holds:
\begin{align*}
\esssup_{(\bQ,\SQ)\in \CPS_{\loc}(\sigma,T)}\bE_\bQ\left[\SQ_T\mid\cF_\sigma\right]=\esssup_{(\bQ,\SQ)\in \CPS_{\loc}(0,T)}\bE_\bQ\left[\SQ_T\mid\cF_\sigma\right].
\end{align*}
\end{theorem} \noindent

\begin{proof}
$1)$ If $(\bQ,\SQ)\in\CPS_\loc(0,T)$, then $(\bQ,\SQ|_{\lsem \sigma,T\rsem})\in\CPS_\loc(\sigma,T)$. Thus, $\CPS_{\loc}(0,T)\subseteq \CPS_{\loc}(\sigma,T)$ we immediately get that
\begin{align*}
\esssup_{(\bQ,\SQ)\in \CPS_{\loc}(\sigma,T)}\bE_\bQ\left[\SQ_T\mid\cF_\sigma\right]\geq \esssup_{(\bQ,\SQ)\in \CPS_{\loc}(0,T)}\bE_\bQ\left[\SQ_T\mid\cF_\sigma\right].
\end{align*}
$2)$ Let $(\lambda_n)_{n\in\bN}\subset (0,1)$ be a sequence such that $\lambda_n\uparrow 1$ as $n$ tends to infinity. Fix an arbitrary $(\bQ,\SQ)\in \SCPS_{\loc}(\sigma,T,\lambda)$ with associated (local) $\bP$-martingales $Z^1$, $Z^2$ as in Proposition \ref{prop:pmartingale}. Then this $(\bQ,\SQ)\in \SCPS_{\loc}(\sigma,T,\lambda)$ can be approximated by a sequence $(\hQn,\hSn)_{n\in\bN}\subset \SCPS_\loc(\sigma,T,\lambda)$ satisfying 
\begin{align}
\label{ali:timeindcond}
(1-\lambda_n)S_\sigma\leq \hSn_\sigma\leq (1+\lambda_n)S_\sigma,	
\end{align}
for each $n\in\bN$ as follows. Define the set $C_\sigma^n\in\cF_\sigma$ by
\begin{align*}
	C_\sigma^n:=\left\{\omega\in\Omega:(1-\lambda_n)S_\sigma(\omega)\leq\SQ_\sigma(\omega)\leq (1+\lambda_n)S_\sigma(\omega)\right\},\quad n\in\bN.
\end{align*}
By Corollary \ref{cor:localcps} there exists for each $n\in\bN$ a consistent local price system $(\bQ^n,\tSn)\in\CPS_\loc(\sigma,T,\lambda_n)$ with associated (local) $\bP$-martingales $Z^{1,n}$, $Z^{2,n}$ such that $\tSn_\sigma=(1+\lambda_n)S_\sigma$. We define $(\hQn,\hSn)\in\SCPS_\loc(\sigma,T,\lambda)$ by its associated (local) $\bP$-martingales by
\begin{align*}
	\widehat{Z}_t^{i,n}:=\mathds{1}_{C_\sigma^n}Z_t^i+\mathds{1}_{(C_\sigma^n)^c}Z_t^{i,n},\quad t\in\lsem\sigma,T\rsem,\ i=1,2,\ n\in\bN.
\end{align*}
Then $(\hQn,\hSn)\in\SCPS_\loc(\sigma,T,\lambda)$ satisfies \eqref{ali:timeindcond} for all $n\in\bN$ and $\widehat{Z}^{i,n}$ converges to $Z^i$, $i=1,2$ for $n$ going to infinity. We now show that for each $n\in\bN$ we can construct a pair $(\bbQ^n,\bSn)\in\CPS_\loc(0,T,\lambda)$ with
\begin{align}
\label{ali:cpscompleteintervalapprox}
	\bE_{\bbQ^n}\left[\bSn_T\mid\cF_\sigma\right]=\bE_{\widehat\bQ^n}\left[\hSn_T\mid\cF_\sigma\right],\quad n\in\bN.
\end{align}
Fix $n\in\bN$. By Lemma \ref{lem:local} and Corollary \ref{cor:localcps} there exists $(\widehat\bQ,\hSQ)\in\CPS_\loc(0,T,\lambda)$ such that  $\hSQ_\sigma=\hSn_\sigma$. Let $\widehat{Z}^1,\widehat{Z}^2$ be the associated (local) $\bP$-martingales as in Proposition \ref{prop:pmartingale}. We define another $\lambda$-consistent local price system $(\bar{\bQ},\bSQ)\in\CPS_\loc(0,T,\lambda)$ by
\begin{align*}
\bar{Z}_t^i:=\begin{cases}
\widehat{Z}_t^i,&\quad 0\leq t\leq \sigma,\\
\widehat{Z}_t^{i,n}\frac{\widehat{Z}_\sigma^i}{\widehat{Z}_\sigma^{i,n}},&\quad \sigma\leq t\leq T,
\end{cases}
\end{align*}
for $i=1,2$. By construction it then holds that $\bSQ_t=\hSn_t$ for all $t\geq \sigma$ and 
\begin{align*}
\bE_{\widehat{\bQ}^n}\left[\hSn_T\mid\cF_\sigma\right]&=\bE_\bP\left[\widehat{Z}_T^{2,n}\mid\cF_\sigma\right]\left(\widehat{Z}_\sigma^{1,n}\right)^{-1}=\bE_\bP\left[\widehat{Z}_T^{2,n}\frac{\widehat{Z}_\sigma^2}{\widehat{Z}_\sigma^{2,n}}\mid\cF_\sigma\right]\left(\widehat{Z}_\sigma^{1,n}\right)^{-1} \frac{\widehat{Z}_\sigma^1}{\widehat{Z}_\sigma^1}\frac{\widehat{Z}_\sigma^{2,n}}{\widehat{Z}_\sigma^2}\\
&=\bE_\bP\left[\bar{Z}_T^2\mid\cF_\sigma\right]\frac{1}{\widehat{Z}_\sigma^{1,n}}\frac{\widehat{Z}_\sigma^1}{\widehat{Z}_\sigma^1}\frac{\widehat{Z}_\sigma^{2,n}}{\widehat{Z}_\sigma^2}=\bE_\bP\left[\bar{Z}_T^2\mid\cF_\sigma\right](\bar{Z}_\sigma^1)^{-1}\frac{\hSn_\sigma}{\hSQ_\sigma}\\
&=\bE_\bP\left[\bar{Z}_T^2\mid\cF_\sigma\right]\left(\bar{Z}_\sigma^1\right)^{-1}
=\bE_{\bar{\bQ}}\left[\bSQ_T\mid\cF_\sigma\right].
\end{align*}
By \eqref{ali:cpscompleteintervalapprox} and because $\widehat Z^{i,n}$ converges to $Z^i$, $i=1,2$, as $n$ tends to infinity, we get that
\begin{align}
\label{ali:timeindapproxcond}
	\esssup_{(\bQ_0,\tilde{S}^{\bQ_0})\in\CPS_\loc(0,T,\lambda)}\bE_{\bQ_0}[\tilde{S}^{\bQ_0}_T\mid\cF_\sigma]\geq\esssup_{n\in \bN}\bE_{\bar\bQ^n}\left[\bSn_T\mid\cF_\sigma\right]\geq \bE_\bQ\left[\SQ_T\mid\cF_\sigma\right].
\end{align}
Since $(\bQ,\SQ)\in\SCPS_\loc(\sigma,T,\lambda)$ was arbitrary, we can take the essential supremum over all $(\bQ,\SQ)\in\SCPS_\loc(\sigma,T,\lambda)$ (resp. $(\bQ,\SQ)\in\CPS_\loc(\sigma,T,\lambda)$) on the right-hand side and conclude that
\begin{align*}
	\esssup_{(\bQ,\SQ)\in\CPS_\loc(0,T,\lambda)}\bE_\bQ\left[\SQ_T\mid\cF_\sigma\right]\geq\esssup_{(\bQ,\SQ)\in\CPS_\loc(\sigma,T,\lambda)}\bE_\bQ\left[\SQ_T\mid\cF_\sigma\right].
\end{align*}
Note that the essential supremum over $\CPS_\loc(\sigma,T,\lambda)$ is equal to the essential supremum over $\SCPS_\loc(\sigma,T,\lambda)$.
\end{proof} \noindent

\subsection{Properties of the fundamental value and asset price bubbles}
\label{subsec:properties}

In this section we study some basic properties of the fundamental value and of asset price bubbles in our setting.

\begin{lemma}
\label{lem:unique}
The fundamental value $F=(F_t)_{t\in [0,T]}$ is an adapted stochastic process, which is unique to within evanescent processes.
\end{lemma} \noindent

\begin{proof}
By Proposition \ref{prop:duality} we obtain that the fundamental value $F$ is given by 
\begin{align*}
	F_t=\esssup_{(\bQ,\SQ)\in\CPS_\loc(0,T,\lambda)}\bE_\bQ\left[\SQ_T\mid\cF_t\right],\quad t\in[0,T].
\end{align*}
By Theorem A.33 of \cite{follmerschied} we obtain that $F$ is an adapted process since all the probability measures $\bQ\in\cQ_\loc$ are equivalent to $\bP$. As Proposition \ref{cor:superrepcond} and Theorem \ref{thm:timeindependent} hold for all stopping times $0\leq \sigma \leq T$, $F$ is unique to within an evanescent set because of the Optional Cross-Section Theorem, see Theorem 86 in Chapter IV of \cite{dellacheriea}.
\end{proof} \noindent

\begin{proposition}
\label{lem:characterization}
Let Assumption \ref{assumption2} hold. Then we have for all stopping times $0\leq \sigma \leq T$ that
\begin{align*}
\esssup_{(\bQ,\SQ)\in\CPS(0,T,\lambda)}\bE_\bQ\left[\SQ_T\mid\cF_\sigma\right]=(1+\lambda)S_\sigma.
\end{align*}
In particular, there is no asset price bubble in the market model.
\end{proposition} \noindent

\begin{proof}
Let $n_0\in\bN$ such that $\frac{1}{n_0}\leq \lambda$. By Assumption \ref{assumption2} there exists a $(\bQ^n,\tSn)\in\CPS(0,T,\frac{1}{n})$ for all $n\in\bN\backslash\{0\}$. Define $\mu_n:=\frac{1+\lambda}{1+\frac{1}{n}}$ for $n\geq n_0$. Then $(\bQ^n,\mu_n\tSn)\in\CPS(0,T,\lambda)$ for all $n\geq n_0$, since for $t\in[0,T]$ we have
\begin{align*}
(1-\lambda)S_t\leq \left(1-\frac{1}{n}\right)S_t\leq \tSn_t\leq \mu_n\tSn_t\leq \mu_n\left(1+\frac{1}{n}\right)S_t=(1+\lambda)S_t.
\end{align*}
Further, it holds
\begin{align*}
(1+\lambda)S_\sigma\geq \esssup_{(\bQ,\SQ)\in\CPS(0,T,\lambda)}\bE_\bQ\left[\SQ_T\mid\cF_\sigma\right]\geq \esssup_{n\geq n_0}\bE_{\bQ^n}\left[\mu_n\tSn_T\mid\cF_\sigma\right]=\esssup_{n\geq n_0}\mu_n\tSn_\sigma,
\end{align*}
where we have used the martingale property of $(\bQ^n,\mu_n\tSn)\in\CPS(0,T,\lambda)$. For the essential supremum we get
\begin{align*}
&\left|(1+\lambda)S_\sigma-\esssup_{n\geq n_0}\mu_n\tSn_\sigma\right|\leq \left|(1+\lambda)S_\sigma- \esssup_{n\geq n_0}\mu_n \left(1-\frac{1}{n}\right)S_\sigma\right|\\=&\left|(1+\lambda)S_\sigma\left(1-\esssup_{n\geq n_0}\frac{1-\frac{1}{n}}{1+\frac{1}{n}}\right)\right|=0.
\end{align*}
Hence we can conclude that
\begin{align*}
\esssup_{(\bQ,\SQ)\in\CPS(0,T,\lambda)}\bE_\bQ\left[\SQ_T\mid\cF_\sigma\right]=(1+\lambda)S_\sigma.
\end{align*}
\end{proof} \noindent

\begin{proposition}
\label{prop:risetransactioncost}
	Let Assumption \ref{assumption} hold. If there exists $\lambda_0\in (0,1)$ such that there is no bubble in the market model with transaction costs $\lambda_0$, then there is no bubble in the market model with transaction costs $\lambda> \lambda_0$.
\end{proposition}

\begin{proof}
	Let $\lambda_0$ be such that 
	\begin{align*}
		F_\sigma^{\lambda_0}=\esssup_{(\bQ,\SQ)\in\CPS_\loc(0,T,\lambda_0)}\bE_\bQ[\SQ_T\mid\cF_\sigma]=(1+\lambda_0)S_\sigma,
	\end{align*}
	for all $[0,T]$-valued stopping times $\sigma$, i.e., there is no bubble in the market model with transaction costs $\lambda_0$. Fix any $\lambda>\lambda_0$. Note that $\CPS_\loc(0,T,\lambda_0)\subseteq\CPS_\loc(0,T,\lambda)$. Then there exists $c>1$ such that $c(1+\lambda_0)=(1+\lambda)$. In particular, if $(\bQ,\SQ)\in\CPS_\loc(0,T,\lambda_0)$, then $(\bQ,c\SQ)\in\CPS_\loc(0,T,\lambda)$. This yields
	\begin{align*}
		F_\sigma^\lambda&=\esssup_{(\bQ,\SQ)\in\CPS_\loc(0,T,\lambda)}\bE_\bQ[\SQ_T\mid\cF_\sigma]\geq \esssup_{(\bQ,\SQ)\in\CPS_\loc(0,T,\lambda_0)}\bE_\bQ[c\SQ_T\mid\cF_\sigma]\\
		&=c(1+\lambda_0)S_\sigma=(1+\lambda)S_\sigma.
	\end{align*}
\end{proof}\noindent
Proposition \ref{prop:risetransactioncost} guarantees that a rise of transaction costs does not yield bubbles' formation.

\begin{lemma}
\label{prop:inclusionmloc}
If the asset price $S=(S_t)_{t\in[0,T]}$ is a semimartingale and the set $\cM_\loc(S)$ of equivalent local martingale measures for $S$ is not empty, then $(\bQ,\mu S)\in\CPS_\loc(0,T)$ for $\bQ\in\cM_\loc(S)\text{ and }\mu\in[1-\lambda,1+\lambda]$, and
\begin{align}
\label{ali:mlocsubset}
F_\sigma=\esssup_{(\bQ,\SQ)\in\CPS_\loc}\bE_\bQ\left[\SQ_T\mid\cF_\sigma\right]\geq \esssup_{\bQ\in\cM_\loc(S)}\bE_\bQ\left[\mu S_T\mid\cF_\sigma\right].
\end{align}
\end{lemma}\noindent

\begin{proof}
	Equation \eqref{ali:mlocsubset} immediately follows by the observation that 
	\begin{align}
		\left\{(\bQ,\mu S):\bQ\in\cM_\loc(S),\mu\in [1-\lambda,1+\lambda]\right\}\subset \CPS_\loc(0,T).
	\end{align}
\end{proof}

\begin{definition}
	Let $D\subseteq \bR$ be an open set in $\bR$. A function $f:D\to\bR$ is said to be upper \textit{semi-continuous} at $x\in D$ if
	\begin{align}
	\label{ali:defuppersemicont}
		\limsup_{y\to x} f(y)\leq f(x).
	\end{align}
	We say that $f$ is upper semi-continuous from the right at $x\in D$, if \eqref{ali:defuppersemicont} holds for $y\downarrow x$. Further, $f$ is called upper semi-continuous (from the right) if $f$ is upper semi-continuous (from the right) for all $x\in D$.
\end{definition}

\begin{theorem}
\label{thm:cadlag}
Suppose that Assumption \ref{assumption} holds and assume that  the function
\begin{align}
\label{ali:uppersemicont}
\varphi(t):= \sup_{(\bQ,\SQ)\in\CPS_\loc}\bE_\bP	\left[\bE_\bQ\left[\SQ_T\mid\cF_t\right]\right],\quad t\in[0,T],
\end{align}
is upper semi-continuous from the right. Then $F=(F_t)_{t\in[0,T]}$ admits a right-continuous modification with respect to $\bP$.
\end{theorem}

\begin{proof}
By Theorem 48 in \cite{dellacherieb}, $F$ admits a right-continuous modification with respect to $\bP$ if and only if for every decreasing sequence $(\beta_n)_{n\in\bN}$ of bounded stopping times $\lim_{n\to\infty}\bE_{\bP}\left[F_{\beta_n}\right]=\bE_{\bP}\left[F_{\lim_{n\to\infty}\beta_n}\right]$. For convenience, we write $\CPS_\loc=\CPS_{\loc}(0,T,\lambda)$ and $\cZ_\loc=\cZ_\loc(0,T,\lambda)$ in the sequel. Note that we use the representation of $(\bQ,\SQ)\in\CPS_\loc$ from Proposition \ref{prop:pmartingale}. As in Proposition 4.3 of \cite{kramkovduality} we first show the identity
\begin{align}
\label{ali:cadlagid}
\bE_{\bP}\left[\esssup_{(\bQ,\SQ)\in\CPS_\loc}\bE_\bQ\left[\SQ_T\mid\cF_\sigma\right]\right]=\sup_{(\bQ,\SQ)\in\CPS}\bE_{\bP}\left[\bE_\bQ\left[\SQ_T\mid\cF_\sigma\right]\right],
\end{align}
for all stopping times $\sigma$ with values in $[0,T]$. For the first direction we use monotonicity to obtain
\begin{align}
\bE_{\bP}\left[\esssup_{(\bQ,\SQ)\in\CPS_\loc}\bE_\bQ\left[\SQ_T\mid\cF_\sigma\right]\right]\geq \sup_{(\bQ,\SQ)\in\CPS_\loc}\bE_{\bP}\left[\bE_\bQ\left[\SQ_T\mid\cF_\sigma\right]\right].
\end{align}
For the reverse direction we use Theorem \ref{thm:timeindependent} to show that $$\Phi:=\left\{\bE_\bQ\left[\SQ_T\mid\cF_\sigma\right]:(\bQ,\SQ)\in\CPS_\loc(\sigma,T)\right\}$$ is directed upwards, i.e. for $\bE_\bQ[\SQ_T\mid\cF_\sigma],\ \bE_{\bar{\bQ}}[\bSQ_T\mid\cF_\sigma]\in\Phi$ there exists $\bE_{\widehat{\bQ}}[\hSQ\mid\cF_\sigma]\in\Phi$ such that $\bE_{\widehat{\bQ}}[\hSQ_T\mid\cF_\sigma]\geq \bE_\bQ[\SQ_T\mid\cF_\sigma]\vee\bE_{\bar{\bQ}}[\bSQ_T\mid\cF_\sigma]$. We define 
\begin{align*}
A_\sigma:=\left\{\bE_\bQ\left[\SQ_T\mid\cF_\sigma\right]\geq \bE_{\bar{\bQ}}\left[\bSQ_T\mid\cF_\sigma\right]\right\}\in\cF_\sigma.
\end{align*}
Let $Z=(Z^1,Z^2)$ and $\bar{Z}=(\bar{Z}^1,\bar{Z}^2)$ be the processes associated to $(\bQ,\SQ)$ and $(\bbQ,\bSQ)$ respectively, as in Proposition \ref{prop:pmartingale}. Then we define
\begin{align}
\frac{d\widehat{\bQ}}{d\bP}=\frac{\widehat{Z}_T^1}{\bE_\bP\left[\widehat{Z}_T^1\right]}:=\frac{\mathds{1}_{A_\sigma} Z_T^1+\mathds{1}_{A_\sigma^c}\bar{Z}_T^1}{\bE_\bP\left[\mathds{1}_{A_\sigma} Z_T^1+\mathds{1}_{A_\sigma^c}\bar{Z}_T^1\right]},
\end{align}
and for $\sigma\leq t\leq T$,
\begin{align}
\widehat{Z}_t^2:=\mathds{1}_{A_\sigma} Z_t^2+\mathds{1}_{A_\sigma^c}\bar{Z}_t^2
\end{align}
with corresponding
\begin{align}
\hSQ_t=\frac{\widehat{Z}_t^2}{\widehat{Z}_t^1}.
\end{align}
Obviously, $\widehat{Z}$ satisfies all requirements from Definition \ref{def:representation}, i.e., $\widehat{Z}\in\mathcal{Z}_\loc$. Clearly, $(1-\lambda)S_t\leq \hSQ_t\leq (1+\lambda)S_t$ for all $t\in[\sigma,T]$. For the local martingale property let $(\tau_n)_{n\in\bN}$ be a localizing sequence for $\SQ$ and $\bSQ$. For $\sigma\leq s\leq t\leq T$ we get
\begin{align*}
&\bE_{\widehat{\bQ}}\left[\left(\hSQ_t\right)^{\tau_n}\mid\cF_s\right]=\bE_\bP\left[\left(\frac{\widehat{Z}_t^2}{\widehat{Z}_t^1}\right)^{\tau_n}\frac{\widehat{Z}_T^1}{\bE_\bP\left[\widehat{Z}_T^1\right]}\mid\cF_s\right]\frac{\bE_\bP\left[\widehat{Z}_T^1\right]}{\widehat{Z}_{s\wedge\tau_n}^1}\\
=&\bE_\bP\left[\left(\mathds{1}_{A_\sigma}Z_t^1+\mathds{1}_{A_\sigma^c}\bar{Z}_t^2\right)^{\tau_n}\mid\cF_s\right]\frac{1}{\widehat{Z}_{s\wedge\tau_n}^1}=\left(\mathds{1}_{A_\sigma}\bE_\bP\left[(Z_t^2)^{\tau_n}\mid\cF_s\right]+\mathds{1}_{A_\sigma^c}\bE_\bP\left[(\bar{Z}_t^2)^{\tau_n}\mid\cF_s\right]\right)\frac{1}{\widehat{Z}_{s\wedge\tau_n}^1}\\
=&\left(\mathds{1}_{A_\sigma}Z_{s\wedge\tau_n}^2+\mathds{1}_{A_\sigma^c}\bar{Z}_{s\wedge\tau_n}^2\right)\frac{1}{\widehat{Z}_{s\wedge\tau_n}^1}=\left(\widehat{S}_s\right)^{\tau_n},
\end{align*}
where we used that $\mathds{1}_{A_\sigma},\mathds{1}_{A_\sigma^c}$ are $\cF_\sigma\subset \cF_s$ measurable. In particular, by Theorem A.33 of \cite{follmerschied}, there exists an increasing sequence $(\bE_{\bQ^n}[\tSn_T\mid\cF_\sigma])_{n\in\bN}\subset\Phi$
\begin{align}
\esssup_{(\bQ,\SQ)\in\CPS_\loc(0,T)}\bE_\bQ\left[\SQ_T\mid\cF_\sigma\right]=\esssup_{(\bQ,\SQ)\in\CPS_\loc(\sigma,T)}\bE_\bQ\left[\SQ_T\mid\cF_\sigma\right]=\lim_{n\to \infty}\bE_{\bQ^n}\left[\tSn_T\mid\cF_\sigma\right].
\end{align}
Thus, we obtain by the Theorem of monotone convergence
\begin{align}
\nonumber
&\bE_{\bP}\left[\esssup_{(\bQ,\SQ)\in\CPS_\loc(0,T,\lambda)}\bE_\bQ\left[\SQ_T\mid\cF_\sigma\right]\right]=\lim_{n\to\infty}\bE_{\bP}\left[\bE_{\bQ^n}\left[\tSn_T\mid\cF_\sigma\right]\right]\\
\label{eq:kramkov1}
\leq &\sup_{(\bQ,\SQ)\in\CPS_\loc(\sigma,T,\lambda)}\bE_{\bP}\left[\bE_\bQ\left[\SQ_T\mid\cF_\sigma\right]\right]=\sup_{(\bQ,\SQ)\in\CPS_\loc(0,T,\lambda)}\bE_{\bP}\left[\bE_\bQ\left[\SQ_T\mid\cF_\sigma\right]\right].
\end{align}
The last equality in \eqref{eq:kramkov1} holds due to similar arguments as in the proof of Theorem \ref{thm:timeindependent}. This concludes the proof of \eqref{ali:cadlagid}.\\
Let now $(\sigma_n)_{n\in\bN}$ be a sequence of stopping times with values in $[0,T]$ such that $\sigma_n\downarrow \sigma$ as $n$ tends to infinity. We now prove that
\begin{align*}
\lim_{n\to\infty} \bE_{\bP}\left[F_{\sigma_n}\right]=\bE_{\bP}\left[\lim_{n\to \infty}F_{\sigma_n}\right]=\bE_{\bP}\left[F_\sigma\right].
\end{align*}
Using \eqref{ali:cadlagid}, the Fatou's lemma, and the fact that $\left(\bE_\bQ[\SQ_T\mid\cF_t]\right)_{\sigma\leq t\leq T}$ is right-continuous we obtain
\begin{align}
\label{ali:cadlagdirection}
\begin{split}
\liminf_{n\to\infty}\sup_{(\bQ,\SQ)\in\CPS_\loc}\bE_{\bP}\left[\bE_\bQ\left[\SQ_T\mid\cF_{\sigma_n}\right]\right]&\geq \liminf_{n\to\infty}\bE_{\bP}\left[\bE_\bQ\left[\SQ_T\mid\cF_{\sigma_n}\right]\right]\\
\geq \bE_{\bP}\left[\liminf_{n\to\infty}\bE_\bQ\left[\SQ_T\mid\cF_{\sigma_n}\right]\right]&=\bE_{\bP}\left[\bE_\bQ\left[\SQ_T\mid\cF_{\sigma}\right]\right].
\end{split}
\end{align}
Since \eqref{ali:cadlagdirection} holds for all $(\bQ,\SQ)\in\CPS_\loc$ we get that
\begin{align}
\label{ali:cadlagdirection1}
\begin{split}
&\liminf_{n\to\infty}\bE_{\bP}\left[F_{\sigma_n}\right]=\liminf_{n\to\infty}\sup_{(\bQ,\SQ)\in\CPS_\loc}\bE_{\bP}\left[\bE_\bQ\left[\SQ_T\mid\cF_{\sigma_n}\right]\right]\\
\geq &\sup_{(\bQ,\SQ)\CPS_\loc}\bE_{\bP}\left[\bE_\bQ\left[\SQ_T\mid\cF_\sigma\right]\right]=\bE_{\bP}\left[F_\sigma\right],
\end{split}
\end{align}
where the last equality follows by \eqref{ali:cadlagid}. By the assumption of upper semi-continuity from the right we directly obtain
\begin{align}
\label{ali:proofuppersemicont}
\limsup_{n\to\infty}\sup_{(\bQ,\SQ)\in\CPS_\loc}\bE_\bP\left[\bE_\bQ\left[\SQ_T\mid\cF_{\sigma_n}\right]\right]\leq\sup_{(\bQ,\SQ)\in\CPS_\loc}\bE_\bP\left[\bE_\bQ\left[\SQ_T\mid\cF_\sigma\right]\right].
\end{align}
Note that \eqref{ali:proofuppersemicont} also implies that the limit is finite, because 
\begin{align*}
	\sup_{(\bQ,\SQ)\in\CPS_\loc}\bE_\bP\left[\bE_\bQ\left[\SQ_T\mid\cF_\sigma\right]\right]\leq \bE_\bP[(1+\lambda)S_\sigma]<\infty.
\end{align*}
Putting \eqref{ali:cadlagdirection1} and \eqref{ali:proofuppersemicont} together yields by \eqref{ali:cadlagid} that
\begin{align*}
	\limsup_{n\to\infty}\bE_\bP[F_{\sigma_n}]=\bE_\bP[F_\sigma]\leq\liminf_{n\to\infty}\bE_\bP[F_n].
\end{align*}
This concludes the proof.

\end{proof}\noindent

\begin{corollary}
\label{cor:cadlag}
Suppose that Assumption \ref{assumption} holds and assume that there exists $\bQ_0\in\cQ_\loc$ such that the function
\begin{align}
\label{ali:uppersemicont}
\varphi(t):= \sup_{(\bQ,\SQ)\in\CPS_\loc}\bE_{\bQ_0}\left[\bE_\bQ\left[\SQ_T\mid\cF_t\right]\right],\quad t\in[0,T],
\end{align}
is upper semi-continuous from the right. Then $F=(F_t)_{t\in[0,T]}$ admits a right-continuous modification with respect to $\bP$.
\end{corollary}

\begin{proof}
	Since $\bP$ and $\bQ_0$ are equivalent we can conclude that if $F$ has a right-continuous modification with respect to $\bQ_0$, $F$ also has a right-continuous modification with respect to $\bP$. 
\end{proof}\noindent
By Proposition \ref{prop:duality}, Lemma \ref{lem:unique} and Theorem \ref{thm:cadlag} we obtain that the fundamental price process is well-defined and admits a right-continuous modification.

\begin{remark}
	The assumption of upper semi-continuity from the right in Theorem \ref{thm:cadlag}  seems rather restrictive. In the frictionless case inequality \eqref{ali:proofuppersemicont} is automatically fulfilled because
	\begin{align*}
		F_t=\esssup_{\bQ\in\cM_\loc(S)}\bE_\bQ[S_T\mid\cF_t],
	\end{align*} 
	for all $t\in[0,T]$ is a $\bQ$-supermartingale for all $\bQ\in\cM_\loc(S)$. Then for any $\bQ_0\in\cM_\loc(S)$ it holds that
	\begin{align*}
		\bE_{\bQ_0}[F_{\sigma_n}]=\sup_{\bQ\in\cM_\loc(S)}\bE_{\bQ_0}\left[\bE_\bQ[S_T\mid\cF_{\sigma_n}]\right]\leq \bE_{\bQ_0}[F_\sigma].
	\end{align*}
	See Proposition 4.3 of \cite{kramkovduality} for further details.\\
	In the presence of transaction costs the supermartingale property of the fundamental value $F$ may fail. Additional regularity conditions on the family of consistent price systems must be required in order to guarantee the existence of a right-continuous modification for $F$.
\end{remark}

\section{Impact of transaction costs on bubbles' formation}
\label{sec:comparison}

In this section we study whether transaction costs can prevent bubbles' formation and their impact on bubbles' size. These issues, if market frictions may prevent bubbles' formation or reduce their impact, has been thoroughly discussed in the economic literature, see the discussion in the Introduction. Here we study these problems from a mathematical point of view in our setting. In particular, we investigate the relation between asset price bubbles in market models with and without transaction costs.\\
To this purpose we use \cite{herdegenschweizer} as reference for the frictionless market case. We now briefly recall and re-adapt the framework of \cite{herdegenschweizer} to be coherent with our setting outlined in Section \ref{sec:setting}. In particular, we assume that the asset price $S$ is given by a c\`{a}dl\`{a}g non-negative semimartingale such that $\cM_\loc(S)\neq \emptyset$. Under these assumptions, NFLVR holds, see \cite{delbaenschachermayer94}. Put $\textbf{S}:=(B,S)$. We denote by $^\sigma L(\textbf{S})$ the set of all $\bR^2$-valued processes $\nu=(\nu_t^1,\nu_t^2)_{\sigma\leq t\leq T}$ which are predictable on $\lsem\sigma,T\rsem$ and for which the stochastic integral process $\int_\sigma^t\nu_sd\textbf{S}_s$, $\sigma\leq t\leq T$, is defined in the sense of $2$-dimensional stochastic integration, see \cite[Section III.6]{protterstochastic}.

\begin{definition}[Definition 2.5, \cite{herdegenschweizer}]
Fix a stopping time $0\leq \sigma\leq T$. The space $^\sigma L^{\text{sf}}(\textbf{S})$ of \emph{self-financing strategies} (for $\textbf{S}$) on $\lsem \sigma,T\rsem$ consists of all $2$-dimensional processes $\nu$ which are predictable on $\lsem \sigma, T\rsem$, belong to $^\sigma L(S)$, and such that the \emph{value process} $V(\nu)(\textbf{S})$ of $\nu$ satisfies the self-financing condition
\begin{align*}
V(\nu)(\textbf{S}):=\nu\cdot \textbf{S}=\nu_\sigma\cdot \textbf{S}_\sigma+\int_\sigma\nu_ud\textbf{S}_u\quad \text{on }\lsem\sigma,T\rsem.
\end{align*}
\end{definition}\noindent

\begin{definition}[Definition 3.1, \cite{herdegenschweizer}]
\label{def:bubbleherdegen}
The \emph{fundamental value} of the asset $S$ at time $t\in[0,T]$ is defined by
\begin{align}
\label{ali:herdegenbubble}
S_t^*:=\essinf\left\{v\in L_+^1(\cF_t,\bP):\exists\nu\in {^*L_+^{\text{sf}}}(\textbf{S}) \text{ with }V_T(\nu)(\textbf{S})\geq S_T\text{ and }V_t(\nu)(\textbf{S})\leq v, \bP\text{-a.s.}\right\}.
\end{align}
We say that the market model has a \emph{strong bubble} if $S^*$ and $S$ are not indistinguishable, i.e., if $\bP(S_\sigma^*<S_\sigma)>0$ for some stopping $0\leq\sigma\leq T$ and define the process $\beta^\NoTC=(\beta_t^\NoTC)_{0\leq t\leq T}$ by $\beta_t^{\NoTC}:=S_t-S_t^*$, $t\in[0,T]$.
\end{definition} \noindent
Note that Definition \ref{def:bubbleherdegen} differs from Definition 3.1 of \cite{herdegenschweizer} since we require $v\in L_+^1(\cF_t,\bP)$ in \eqref{ali:herdegenbubble} to be consistent with Definition \ref{def:bubble}. In this setting the duality from Theorem 3.2 of \cite{kramkovduality} holds and we get
\begin{align}
\label{app:kramkovduality}
S_\sigma^*=\esssup_{\bQ\in\cM_\loc(S)}\bE_\bQ\left[S_T\mid\cF_\sigma\right].
\end{align}
In the more general framework of \cite{herdegenschweizer} it is possible that the duality does not hold, see Remark 3.11 of \cite{herdegenschweizer} and the comment before for more information.\\
Recall that, in the setting of \cite{herdegenschweizer} we assume that the asset price $S$ is a semimartingale and that $\cM_\loc(S)\neq \emptyset$. At time $t\in[0,T]$ we obtain by Lemma \ref{prop:inclusionmloc} that
\begin{align*}
(1+\lambda)S_t\geq F_t=\esssup_{(\bQ,\tilde{S}^Q)\in\CPS_\loc}\bE_\bQ\left[\tilde{S}_T^Q\mid\cF_t\right]\geq \esssup_{\bQ\in\cM_{loc}(S)}\bE_\bQ[(1+\lambda)S_T\mid\cF_t]=(1+\lambda)S_t^*.
\end{align*}
In particular, we have
\begin{align}
\label{ali:bubblerelation}
\beta_t=(1+\lambda)S_t-F_t\leq (1+\lambda)\left(S_t-\esssup_{\bQ\in\cM_{loc}(S)}\bE_\bQ[S_T\mid\cF_t]\right)=(1+\lambda)\beta_t^{\NoTC}.
\end{align}
From \eqref{ali:bubblerelation} we immediately obtain for $t\in[0,T]$ that if $\beta_t^\NoTC=0$, then $\beta_t=0$ and that $\beta_t^\NoTC>0$ if $\beta_t>0$ for $t\in[0,T]$. Let $\beta_t^\NoTC\neq 0$. Then
\begin{align}
\label{ali:bubblebound}
\frac{\beta_t}{\beta_t^{\NoTC}}\leq 1+\lambda,
\end{align}
which means that the quotient of the bubbles is bounded by the factor $(1+\lambda)$. Furthermore, we have
\begin{align}
\label{ali:differencebubble1}
-\lambda\beta^\NoTC\leq \beta_t^\NoTC-\beta_t\leq \beta_t^\NoTC.
\end{align}
It is easy to see that both bounds in \eqref{ali:differencebubble1} can be obtained. In Example \ref{ex:bubblebirth} we get

\begin{align}
\label{ali:bubblereductionexample}
\beta_t-\beta_t^\NoTC=(1+\lambda)S_t\mathds{1}_{\{\gamma\leq t\}}-S_t\mathds{1}_{\{\gamma\leq t\}}=\lambda S_t\mathds{1}_{\{\gamma\leq t\}}=\lambda\beta_t^\NoTC.
\end{align}
Furthermore, we have in Example \ref{examplebessel} that $\beta_t\equiv 0$ and thus we obtain the equality on the right hand side of \eqref{ali:differencebubble1}. Further, we note that

\begin{align}
\label{ali:bubblesize}
\nonumber
&(1+\lambda)\beta_t^{\NoTC}-\beta_t\\ \nonumber=&(1+\lambda
)\left(S_t-\esssup_{\bQ\in\cM_{loc}(S)}\bE_\bQ[S_T\mid\cF_t]\right)-(1+\lambda)S_t+\esssup_{(\bQ,\tilde{S}^Q)\in\CPS_\loc}\bE_\bQ\left[\tilde{S}_T^Q\mid\cF_t\right]\\
=&\esssup_{(\bQ,\tilde{S}^Q)\in\CPS_\loc}\bE_\bQ\left[\tilde{S}_T^Q\mid\cF_t\right]-(1+\lambda)\esssup_{\bQ\in\cM_{loc}(S)}\bE_\bQ[S_T\mid\cF_t]=:\Delta_{t,T}(\lambda).
\end{align}
By rearranging equation \eqref{ali:bubblesize} we obtain then
\begin{align}
\label{ali:bubblesize2}
\beta_t=(1+\lambda)\beta_t^{\NoTC}-\Delta_{t,T}(\lambda).
\end{align}
Clearly, it holds $\Delta_{t,T}(\lambda)\in \lsem 0,(1+\lambda)\beta_t^\NoTC\rsem$. Consider Example \ref{examplebessel}, where $S$ is a $3$-dimensional inverse Bessel process with respect to $\bP$ and set $t=0$. Then $\beta_0=0$ and
\begin{align}
\begin{split}
\label{ali:besseldifference}
\Delta_{t,T}(\lambda)&=\sup_{(\bQ,\SQ)\in\CPS_\loc}\bE_\bQ\left[\SQ_T\right]-(1+\lambda)\sup_{\bQ\in\cM_\loc(S)}\bE_\bQ\left[S_T\right]\\
&=(1+\lambda)S_0-(1+\lambda)\bE_\bP\left[S_T\right]\\
&=2(1+\lambda)\left(1-\Phi\left(\frac{1}{\sqrt{T}}\right)\right),
\end{split}
\end{align}
where $\Phi$ denotes the cumulative distribution function of the standard normal distribution, see \cite{follmerprotter}. For $T$ tending to infinity, then $\Delta_{t,T}(\lambda)$ tends to $(1+\lambda)$.

\begin{remark}
From equation \eqref{ali:bubblerelation} we can see that if a market model without transaction costs has no asset price bubble, then the analogue market model with transaction has no asset price bubble either. In other words, the introduction of transaction costs into a market cannot generate asset price bubbles. Conversely, by \eqref{ali:bubblerelation} it follows that, if a market model with transaction costs has an asset price bubble, the corresponding frictionless market model has an asset price bubble as well. \\
In our model the introduction of transaction costs can possibly prevent the occurrence of an asset price bubble. This can be seen in Example \ref{examplebessel} where we have an asset price bubble in the sense of Definition \ref{def:bubbleherdegen} but no bubble in presence of transaction costs with respect to Definition \ref{def:bubble}. However, Example \ref{ex:bubblebirth} shows that it is possible to have an asset price bubble in both market models, with and without transaction costs. In particular, the presence of transaction costs does not guarantee the absence of asset price bubbles.
\end{remark}

\section{Examples}
\label{sec:examples}

In this section we provide several examples to illustrate our setting and the impact of transaction costs on asset bubbles. In Example \ref{examplefrac} we start by showing a market model under transaction cost where the asset price, driven by a fractional Brownian motion, has a bubble in the sense of Definition \ref{def:bubble}. Due to well-known no-arbitrage arguments, a process driven by the fractional Brownian motion is not admissible as price process in a frictionless market model. \\
Then we study how the presence of an asset price bubble in a market model without transaction costs may be related to the appearance of a bubble in the analogous market model with transaction costs, and vice versa. To this purpose, we consider examples where the asset price is a semimartingale in the framework of \cite{herdegenschweizer} for frictionless market models. In Example \ref{examplemartingale} we illustrate a standard market model such that there is no bubble, neither with nor without transaction cost. This example suggests that the introduction of transaction costs cannot lead to bubble's formation. In Example \ref{examplebessel}, the market model has no bubble under transaction cost but there is a bubble in the frictionless market model in the sense of \cite{herdegenschweizer}. In particular, this shows how transaction costs can possibly prevent the appearance of an asset price bubble. Example \ref{examplemartingale} and \ref{examplebessel} illustrate the impact of transaction costs on bubbles' formation. In Example \ref{ex:bubblebirth} we see how bubble's birth is already included in our model.

\begin{example}
\label{examplefrac}
This example is based on Example 7.1 of \cite{guasonifragility}. Let $W^H$ be a fractional Brownian motion with Hurst index $0<H<1$. We define

\begin{align*}
X_t:=\exp(W_t^H+\mu t),\quad t\geq 0,
\end{align*}
for $\mu\geq 0$. Let $\bF^X:=(\cF_t^X)_{t\geq 0}$ be the (completed) natural filtration of the process $X$. Note that $X$ admits a consistent price system in the non-local sense on the interval $[0,T]$ for all $T>0$ by Proposition 4.2 of \cite{schachermayercps}. Define the stopping time

\begin{align*}
\tau:=\inf\left\{t\in\bR:X_t=\frac{1}{2}\right\},
\end{align*}
and set

\begin{align*}
S_t:=X_{\tau\wedge \tan t},\quad 0\leq t<\frac{\pi}{2},\quad S_t=\frac{1}{2},\quad t\geq \frac{\pi}{2}.
\end{align*}
Define $\cG_t:=\cF_{\tan t}$, $0\leq t<\pi/2$, and $\cG_{\pi/2}:=\cF_\infty$. Consider $T\geq \pi/2$. We now show that there exists no consistent price system in the non-local sense for any $\lambda\in (0,1)$. By contradiction assume that there exists a consistent price system $(\bQ,\SQ)$ for $S$ in the non-local sense for a $\lambda\in (0,1)$. Then we have

\begin{align}
\frac{1-\lambda}{2}\leq \SQ_t=\bE_\bQ\left[\SQ_T\mid\cF_t\right]\leq \frac{1+\lambda}{2}\ a.s.
\end{align}
for all $0\leq t\leq T$, and hence also

\begin{align}
\frac{1-\lambda}{2(1+\lambda)}\leq S_t\leq \frac{1+\lambda}{2(1-\lambda)},
\end{align}
which is not possible because $S$ is not bounded from above for $0<t<\pi/2$. Thus, we can conclude that there is no consistent price system in the non-local sense. However, $S$ satisfies Assumption \ref{assumption}, i.e., for every $\lambda>0$ there exists a consistent local price system for $S$. Since $X$ admits for all $\lambda>0$ a consistent local price system $(\bQ,\SQ)$ on $[0,T]$ for all $T>0$, $(\bQ,(\SQ)^\tau)$ is also a consistent local price system for $S$ on $[0,T]$. We now show that there is a bubble in this market model with transaction costs for $\lambda<1/3$. For any consistent local price system $(\bQ,\SQ)$ for $S$ we have
\begin{align*}
(1-\lambda)S_0\leq \SQ_0\leq(1+\lambda)S_0,
\end{align*}
and because $S_0=1$
\begin{align*}
\frac{1-\lambda}{2}\leq \SQ_T\leq \frac{1+\lambda}{2}.
\end{align*}
This implies for $\lambda< 1/3$
\begin{align}
\label{ali:examplebubble}
\SQ_0\geq 1-\lambda >\frac{1+\lambda}{2}\geq \SQ_T
\end{align}
for all consistent local price systems. Thus, we have
\begin{align}
\label{ali:examplebubble2}
 (1+\lambda)S_0\geq \esssup_{(\bQ,\SQ)\in\CPS_\loc}\SQ_0>\esssup_{(\bQ,\SQ)\in\CPS_\loc}\bE_\bQ[\SQ_T].
\end{align} 
Therefore, we can conclude by equation \eqref{ali:examplebubble2} that the the asset $S$ has a bubble under transaction cost at time $t=0$.
\end{example}

\begin{remark}
\label{rem:example}
Due to the well-known arbitrage arguments, see \cite{delbaenschachermayer94}, the process $X$ in Example \ref{examplefrac} cannot be considered to describe asset price dynamics in a market model without transaction costs. Hence in the case a comparison with an analogous frictionless market model makes no-sense. Note also that the process $X$ can be replaced by any c\`adl\`ag process which is not bounded and admits a consistent local price system on $[0,T]$ for all $T>0$ and for all $\lambda>0$.
\end{remark}\noindent

\begin{example}
\label{examplemartingale}
Let $S$ be a true $\bQ_0$-martingale for some probability measure $\bQ_0\sim \bP$. Then, $\tilde{S}^{\bQ_0}:=((1+\lambda)S_t)_{0\leq t\leq T}$ is a true $\bQ_0$-martingale and $(\bQ_0,\tilde{S}^{\bQ_0})$ is a consistent price system in the non-local sense for $S$. For any stopping time $0\leq\tau\leq T$ we obtain by Proposition \ref{prop:fundamentalbound} that

\begin{align}
(1+\lambda)S_\tau\geq \esssup_{(\bQ,\tilde{S}^Q)\in\CPS_\loc}\bE_\bQ\left[\tilde{S}_T^\bQ\mid\cF_\tau\right]\geq \bE_{\bQ_0}\left[(1+\lambda)S_T\mid\cF_\tau\right]=(1+\lambda)S_\tau=(1+\lambda)S_\tau.
\end{align}
Hence there is no bubble in the market model with transaction costs. Alternatively, we can observe that Assumption \ref{assumption2} is satisfied and thus we can apply Proposition \ref{lem:characterization}. \\
From Definition \ref{def:bubbleherdegen} we have for any stopping time $0\leq \tau\leq T$

\begin{align*}
S_\tau\geq \esssup_{\bQ\in\cM_{loc}(S)}\bE_\bQ\left[S_T\mid\cF_\tau\right]\geq \bE_{\bQ_0}\left[S_T\mid\cF_\tau\right]=S_\tau.
\end{align*}
So there is also no bubble in the market model without transaction cost in the sense of \cite{herdegenschweizer}.
\end{example}

\begin{example}
\label{examplebessel}
In this example we assume that $S$ is given by a three-dimensional inverse Bessel process, i.e.,
\begin{align}
S_t:=\|B_t\|^{-1},\quad t\in[0,T],
\end{align}
where $(B_t)_{t\in[0,T]}=(B_t^1,B_t^2,B_t^3)_{t\in[0,T]}$ is a three-dimensional Brownian motion with $B_0=(1,0,0)$ and consider the filtration $\bF^S$ defined by $\cF_t^S:=\sigma(S_s:s\leq t)$. Example 5.2 in \cite{herdegenschweizer}, shows that there is a bubble in the market model without transaction cost in the sense of Definition \ref{def:bubbleherdegen}. Note that there is also a $\bP$-bubble in the sense of \cite{protterbubbles} as in the case of a complete market model these definitions coincide. However, by Theorem 5.2 of \cite{guasonifragility} we have that for all $\lambda>0$ there exists $(\bQ,\SQ)\in\CPS$, where $\SQ$ is a true $\bQ$-martingale such that

\begin{align}
\label{ali:bessel1}
(1-\lambda)S_t\leq \SQ_t\leq (1+\lambda)S_t,\quad\text{for all } t\in [0,T].
\end{align}	
In the notation of \cite{guasonifragility}, we say that $\SQ$ is $\lambda$-close to $S$. 
In particular, Assumption \ref{assumption2} is satisfied and thus we obtain by Proposition \ref{lem:characterization} that there is no bubble in any market model with proportional transaction costs $\lambda>0$. This example shows that proportional transaction costs can prevent bubbles' formation.
\end{example}

\begin{example}
\label{ex:bubblebirth}
This example is based on Example 5.4 of \cite{herdegenschweizer}. It illustrates that bubble birth (see \cite{protterbubbles}, \cite{biagini2014shifting}) is naturally included in our model. \\
Let $\gamma$ be a random variable with values in $(0,1]$, $0<\bP(\gamma=1)<1$ and $\bP(\gamma\geq t_0>0)=1$ for some $t_0\in(0,1)$ and consider the filtration $\bF^\gamma$ generated by $H_t=\mathds{1}_{\{\gamma\leq t\}},\ t\in[0,1]$. Then $\gamma$ is a $\bF^\gamma$ stopping time, which represents the time when the bubble is born. Further, let $W$ be a Brownian motion independent of $\gamma$. Denote by $\bF^W$ the natural filtration generated by $W$ and define the filtration $\bF=(\cF_t)_{t\in[0,1]}$ by $\cF_t:=\cF_t^W\vee \cF_t^\gamma\vee \cN$, $t\in[0,1]$, where $\cN$ denotes the $\bP$-nullsets of $\cF_1^W\vee \cF_1^\gamma$. Then $\gamma$ is also an $\bF$-stopping time. Let $S=(S_t)_{0\leq t\leq 1}$ be the unique strong solution to the SDE
\begin{align}
\label{ali:example55S}
dS_t=S_t \left(\mu\d t+v(t,\gamma)\d W_t\right),\quad S_0=1,
\end{align}
with $\mu\in\bR$ and $v:[0,1]^2\to [v_0,\infty)$ given by
\begin{align}
v(t,u)=v_0\left(1+\frac{1}{1-t}\mathds{1}_{\{u\leq t<1\}}\right),
\end{align}
for $v_0>0$. Then $S$ is a geometric Brownian motion up to $\gamma$. At time $\gamma$ the term $1/(1-t)$ starts to influence the volatility which explodes until time $1$. This implies that $S$ converges to $0$ as $t$ tends to $1$. We determine the fundamental value $F_\sigma$ of $S$ at time $\sigma<1$. In particular, we see that there is no bubble before time $\gamma$ but the bubble starts at $\gamma$. The fundamental value $F_\sigma$ at time $\sigma<1$ is given by
\begin{align}
\label{ali:bubblebirthfunda}
F_\sigma=(1+\lambda)S_\sigma\mathds{1}_{\{\gamma>\sigma\}}.
\end{align}
Note that $S_1(\omega_0)=0$ for $\omega_0\in \{\omega\in\Omega: \gamma(\omega)\leq \sigma(\omega)\}$. We define the strategy $\varphi=(\varphi_t^1,\varphi_t^2)_{t\in\lsem\sigma,T\rsem}$ on $\lsem \sigma,T\rsem$ by
\begin{align*}
(\varphi_t^1,\varphi_t^2)=\begin{cases}
\left((1+\lambda)S_\sigma\mathds{1}_{\{\gamma>\sigma\}},0\right),&\quad\text{ for }t=\sigma,\\
\left(0,\mathds{1}_{\{\gamma>\sigma\}}\right),&\quad \text{ for }\sigma<t<1,\\
(0,1),&\quad\text{ for } t=1.
\end{cases}
\end{align*} 
That is, using the initial capital $(1+\lambda)S_\sigma\mathds{1}_{\{\gamma>\sigma\}}$ we trade in such a way that we hold the asset at time $1$. If at time $\sigma$, $\gamma$ has already happened, we know that the volatility blows up and we can buy the asset at time $1$ at price $0$. However, if $\gamma$ happens strictly after $\sigma$ we do not know if the volatility will blow up and thus we buy the asset at time $\sigma$ at price $(1+\lambda)S_\sigma$ in order to hold the asset at time $1$. As this strategy $\varphi$ super-replicates the position $(0,1)$, we conclude that $F_\sigma\leq (1+\lambda)S_\sigma\mathds{1}_{\{\gamma>\sigma\}}$.\\
For the reverse direction, ``$\geq$'' we use the duality from Proposition \ref{cor:superrepcond}. By Example 5.4 of \cite{herdegenschweizer} we get
\begin{align*}
\esssup_{\bQ\in\cM_\loc(S)}\bE_\bQ\left[S_1\mid\cF_\sigma\right]=S_\sigma\mathds{1}_{\{\gamma >\sigma\}}.
\end{align*}
From this we obtain
\begin{align*}
F_\sigma=\esssup_{(\bQ,\SQ)\in\CPS_\loc}\bE_\bQ\left[\SQ_1\mid\cF_\sigma\right]\geq \esssup_{\bQ\in\cM_\loc(S)}\bE_\bQ\left[(1+\lambda)S_1\mid\cF_\sigma\right]=(1+\lambda)S_\sigma\mathds{1}_{\{\gamma >\sigma\}}.
\end{align*}
Indeed, we have
\begin{align*}
F_\sigma=(1+\lambda)S_\sigma\mathds{1}_{\{\gamma>\sigma\}}.
\end{align*}
This implies that $F_\sigma=(1+\lambda)S_\sigma$ on $\{\sigma<\gamma\}$ and $F_\sigma=0$ on $\{\sigma\geq \gamma\}$. In particular, we can conclude that $\gamma$ is then the time at which the bubble is born.\\
We illustrate this example in Figure \ref{simulation} below. Before the stopping time $\gamma$ occurs the price process and the fundamental value are equal. The fundamental value drops immediately to $0$ when $\gamma$ happens.

\begin{figure}[h]
\centering
\includegraphics[width=1.0\linewidth,height=0.5\textheight]{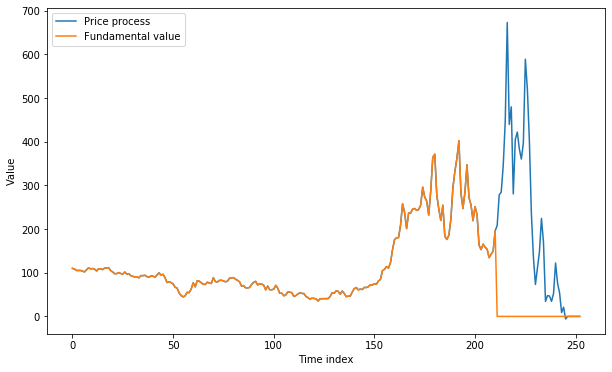}
\caption{Example \ref{ex:bubblebirth}: Simulation of $253$ days of $S$, defined in \eqref{ali:example55S} with $\mu=0.3$, $\sigma_0= 0.4$ and the starting time $\gamma$ of the bubble being uniformly distributed on $(0,1)$. Before $\gamma$ the fundamental value coincides with $(1+\lambda)S$. When $\gamma$ occurs, the fundamental value drops to $0$.}
\label{simulation}
\end{figure}

\end{example}

%

\appendix

\section{Super-replication Theorems}
\label{app:superreplication}
For sake of completeness we provide the super-replication theorems (Theorem 1.4, Theorem 1.5) of \cite{schachermayersuperreplication}. Note that Theorem 1.5 of \cite{schachermayersuperreplication} coincides with Theorem 4.1 of \cite{campischachermayer}. 

\begin{theorem}[Theorem 1.5, \cite{schachermayersuperreplication}]
\label{thm:superrep}
Let Assumption \ref{assumption2} hold. We consider a contingent claim $X_T=(X_T^1,X_T^2)$ which pays $X_T^1$ many units of the bond and $X_T^2$ many units of the risky asset at time $T$. The random variable $X_T$ is assumed to satisfy
\begin{align*}
X_T^1+\left(X_T^2\right)^+(1-\lambda)S_T-\left(X_T^2\right)^-(1+\lambda)S_T\geq -M(1+S_T),
\end{align*}
for some $M>0$. Then the following assertions are equivalent
\begin{enumerate}
\item There is a self-financing trading strategy $\varphi$ with $\varphi\equiv 0$ on $\lsem 0,\sigma\rsem$ and $\varphi_T=(X_T^1,X_T^2)$ which is admissible in a num\'eraire-free sense.
\item For every $(\bQ,\SQ)\in\CPS(\sigma,T)$ we have
\begin{align}
\label{eq:superrepupperbound}
\bE_\bQ\left[X_T^1+X_T^2\SQ_T\right]\leq 0.
\end{align}
\end{enumerate}
\end{theorem} \noindent

\begin{theorem}[Theorem 1.4, \cite{schachermayersuperreplication}]
\label{thm:superreploc}
Let Assumption \ref{assumption} hold. We consider a contingent claim $X_T=(X_T^1,X_T^2)$ which pays $X_T^1$ many units of the bond and $X_T^2$ many units of the risky asset at time $T$. The random variable $X_T$ is assumed to satisfy
\begin{align}
\label{eq:condsuperreploc}
X_T^1+\left(X_T^2\right)^+(1-\lambda)S_T-\left(X_T^2\right)^-(1+\lambda)S_T\geq -M,
\end{align}
for some $M>0$. Then the following assertions are equivalent:
\begin{enumerate}
\item\label{item:superreploc1} There is a self-financing trading strategy $\varphi$ on $\lsem\sigma,T\rsem$ with $\varphi\equiv 0$ on $\lsem 0,\sigma\rsem$ and $\varphi_T=(X_T^1,X_T^2)$ which is admissible in a num\'eraire-based sense.
\item\label{item:superreploc2} For every $(\bQ,\SQ)\in\CPS_\loc(\sigma,T)$ we have
\begin{align}
\label{eq:superrepupperboundloc}
\bE_\bQ\left[X_T^1+X_T^2\SQ_T\right]\leq 0.
\end{align}
\end{enumerate}
\end{theorem} \noindent
Note that in Theorem \ref{thm:superreploc} we consider the claim $X_T=(X_T^1,X_T^2)$ instead of $X_T=(X_T^1,0)$ as in Theorem 1.4 of \cite{schachermayersuperreplication}. However, the proof is similar. For details on the proof, we refer to \cite{thesis}.

\section*{Acknowledgement}
We would like to thank Martin Schweizer for helpful remarks and insightful discussions. 

\bibliography{Bubbles_TC_Paperversion}

\bibliographystyle{acm}

\end{document}